\documentclass[journal, twocolumn]{IEEEtran}
\usepackage{amssymb}
\usepackage{amsmath}
\usepackage{amsthm}
\usepackage{amsfonts}
\usepackage{graphicx}
\usepackage[numbers, square, comma, sort&compress]{natbib}
\usepackage{xcolor}
\usepackage[nolist]{acronym}
\usepackage[bookmarks=false]{hyperref}
\usepackage{caption}
\usepackage{subcaption}
\usepackage{comment}
\usepackage{csquotes}
\usepackage{tikz}
\usepackage{bm}
\usetikzlibrary{arrows,shapes.misc,chains,scopes}
\usepackage{pgfplotstable}
\pgfplotsset{compat=1.15}
\usepackage[ruled,linesnumbered,vlined]{algorithm2e}
\SetNlSty{bfseries}{\color{black}}{}
\SetArgSty{textnormal}
\SetKw{Continue}{continue}
\usepackage{enumitem}
\setlist[itemize]{leftmargin=*}
\SetKw{Break}{break}

\newcommand{\noiseguess}[1]{z^{n,{#1}}}
\newcommand{\rvnoise}{N^n}
\newcommand{\noise}{z^n}
\newcommand{\rvchanout}{R^n}
\newcommand{\rvchanouti}{R_i}
\newcommand{\chanout}{r^n}
\newcommand{\chanouti}{r_i}
\newcommand{\rvdemodout}{Y^n}
\newcommand{\demodout}{y^n}
\newcommand{\codebook}{\mathcal{C}}
\newcommand{\codeword}{c^n}
\newcommand{\codewordi}{c_i}
\newcommand{\cleaned}{\hat{x}^n}
\newcommand{\rvword}{X^n}

\newcommand{\binaryAlphabetN}{\{0,1\}^n}

\newcommand{\guessworkFun}{G}
\newcommand{\guesswork}{\guessworkFun(\rvnoise)}

\newcommand{\mlcodeword}{c^{n,*}}

\newcommand{\decodelist}{\mathcal{L}}

\def\A{ A }
\def  \P{{P}}
\def\PA{{\P(A)}}
\def\PHitIncorrect{\phi}
\newcommand{\B}[1]{{B_{#1}}}
\newcommand{\PB}[1]{{\P(B_{#1})}}

\newcommand{\W}[1]{{U_{(#1)}}}
\newcommand{\Wq}[2]{{U_{(#1), #2}}}
\newcommand{\Wij}[2]{{U_{(#1)}^{#2}}}
\newcommand{\qij}[2]{{q_{#1}^{#2}}}
\newcommand{\qiLj}[1]{{q_1^{L,\{#1\}}}}

\newtheorem{theorem}{Theorem}
\newtheorem{corollary}{Corollary}

\def\Fb{\mathbb{F}}

\graphicspath{{./img/}}

\title{Soft-output (SO) GRAND and Iterative Decoding \\to Outperform LDPCs}

\begin{document}

\author{Peihong~Yuan,~\IEEEmembership{Member,~IEEE,}
        Muriel M\'edard,~\IEEEmembership{Fellow,~IEEE,}
        Kevin Galligan,
       and~Ken~R.~Duffy,~\IEEEmembership{Senior Member,~IEEE}
\thanks{This paper was presented in part at 2023 IEEE Globecom~\cite{galligan2023upgrade}.}
\thanks{P. Yuan and M. M{\'e}dard are with the Massachusetts Institute of Technology Network Coding \& Reliable Communications Group (e-mails: \{phyuan, medard\}@mit.edu).}
\thanks{K. Galligan is with the Maynooth University Hamilton Institute, (e-mail: kevin.galligan.2020@mumail.ie).}
\thanks{K. R. Duffy is with the Northeastern University Engineering Probability Information \& Communications Laboratory (e-mail: k.duffy@northeastern.edu).}
\thanks{This work was supported by the Defense Advanced Research Projects Agency (DARPA) under Grant HR00112120008. \emph{(Corresponding author: Muriel M{\'e}dard)}}

}

\begin{acronym}
    \acro{PC}{product code}
    \acro{BCH}{Bose-Chaudhuri-Hocquenghem}
    \acro{eBCH}{extended Bose-Chaudhuri-Hocquenghem}
    \acro{GRAND}{guessing random additive noise decoding}
    \acro{SGRAND}{soft GRAND}
    \acro{DSGRAND}{discretized soft GRAND}
    \acro{SRGRAND}{symbol reliability GRAND}
    \acro{ORBGRAND}{ordered reliability bits GRAND}
    \acro{SOGRAND}{soft-output GRAND}
    \acro{5G}{the $5$-th generation wireless system}
    \acro{NR}{new radio}
	\acro{APP}{a-posteriori probability}
	\acro{ARQ}{automated repeat request}
	\acro{ASK}{amplitude-shift keying}
	\acro{AWGN}{additive white Gaussian noise}
	\acro{B-DMC}{binary-input discrete memoryless channel}
	\acro{BEC}{binary erasure channel}
	\acro{BER}{bit error rate}
	\acro{biAWGN}{binary-input additive white Gaussian noise}
	\acro{BLER}{block error rate}
	\acro{uBLER}{undetected block error rate}
	\acro{bpcu}{bits per channel use}
	\acro{BPSK}{binary phase-shift keying}
	\acro{BSC}{binary symmetric channel}
	\acro{BSS}{binary symmetric source}
	\acro{CDF}{cumulative distribution function}
	\acro{CRC}{cyclic redundancy check}
	\acro{DE}{density evolution}
	\acro{DMC}{discrete memoryless channel}
	\acro{DMS}{discrete memoryless source}
	\acro{BMS}{binary input memoryless symmetric}
	\acro{eMBB}{enhanced mobile broadband}
	\acro{FER}{frame error rate}
	\acro{uFER}{undetected frame error rate}
	\acro{FHT}{fast Hadamard transform}
	\acro{GF}{Galois field}
	\acro{HARQ}{hybrid automated repeat request}
	\acro{i.i.d.}{independent and identically distributed}
	\acro{LDPC}{low-density parity-check}
	\acro{GLDPC}{generalized low-density parity-check}
	\acro{LHS}{left hand side}
	\acro{LLR}{log-likelihood ratio}
	\acro{MAP}{maximum-a-posteriori}
	\acro{MC}{Monte Carlo}
	\acro{ML}{maximum-likelihood}
	\acro{PDF}{probability density function}
	\acro{PMF}{probability mass function}
	\acro{QAM}{quadrature amplitude modulation}
	\acro{QPSK}{quadrature phase-shift keying}
	\acro{RCU}{random-coding union}
	\acro{RHS}{right hand side}
	\acro{RM}{Reed-Muller}
	\acro{RV}{random variable}
	\acro{RS}{Reed–Solomon}
	\acro{SCL}{successive cancellation list}
	\acro{SE}{spectral efficiency}
	\acro{SNR}{signal-to-noise ratio}
	\acro{UB}{union bound}
	\acro{BP}{belief propagation}
	\acro{NR}{new radio}
    \acro{CA-Polar}{CRC-assisted successive cancellation list}
	\acro{CA-SCL}{CRC-assisted successive cancellation list}
	\acro{DP}{dynamic programming}
	\acro{URLLC}{ultra-reliable low-latency communication}
    \acro{GCC}{Generalized Concatenated Code}
    \acro{GCCs}{Generalized Concatenated Codes}
    \acro{MDS}{maximum distance separable}
    \acro{ORBGRAND-AI}{ORBGRAND Approximate Independence}
    \acro{BLER}{block error rate}
    \acro{crc}[CRC]{cyclic redundancy check}
    \acro{cascl}[CA-SCL]{CRC-Assisted Successive Cancellation List}
    \acro{uer}[UER]{undetected error rate}
    \acro{MDR}{misdetection rate}
    \acro{LMDR}{list misdetection rate}
    \acro{QC}{quasi-cyclic}
    \acro{SISO}{soft-input soft-output}
    \acro{VN}{variable node}
    \acro{CN}{check node}
    \acro{SI}{soft-input}
    \acro{SO}{soft-output}
    \acro{PAC}{polarization-adjusted convolutional}
    \acro{dRM}{dynamic Reed Muller}
    \acro{QPSK}{quadrature phase shift keying}
    \acro{MCS}{modulation and coding scheme}
    \acro{IR-HARQ}{incremental redundancy hybrid automatic repeat request}
\end{acronym}
\maketitle

\begin{abstract}

We establish that a large, flexible class of long, high redundancy error correcting codes can be efficiently and accurately decoded with guessing random additive noise decoding (GRAND). Performance evaluation demonstrates that it is possible to construct simple product codes with lengths of approximately $200$ to $4000$ bits and rates between $0.2$ and $0.8$ that outperform low-density parity-check (LDPC) codes from the 5G New Radio standard in both AWGN and fading channels. The concatenated structure enables many desirable features, including: low-complexity hardware-friendly encoding and decoding; significant flexibility in length and rate through modularity; and high levels of parallelism in encoding and decoding that enable low latency.

Central is the development of a method through which any soft-input (SI) GRAND algorithm can provide soft-output (SO) in the form of an accurate a-posteriori estimate of the likelihood that a decoding is correct or, in the case of list decoding, the likelihood that each element of the list is correct. The distinguishing feature of soft-output GRAND (SOGRAND) is the provision of an estimate that the correct decoding has not been found, even when providing a single decoding. That per-block SO can be converted into accurate per-bit SO by a weighted sum that includes a term for the SI. Implementing SOGRAND adds negligible computation and memory to the existing decoding process, and using it results in a practical, low-latency alternative to LDPC codes.
\end{abstract}


\section{Introduction}

Since Shannon's pioneering work \cite{shannon1948}, the quest to design efficiently decodable long, high-redundancy error correction codes has resulted in technical revolutions from product codes \cite{elias_error-free_1954} to concatenated codes \cite{forney1966_concatenated} to turbo codes \cite{berrou_1993_turbo} to turbo product codes \cite{pyndiah_1998} to \ac{LDPC} codes \cite{gallagher1962_lowdensity, costello2007channel}. Since the rediscovery of \ac{LDPC} codes in the 1990s \cite{Sipser96,mackay1997near,mackay2003information} and the subsequent development of practical encoders and decoders, e.g. \cite{richardson2001efficient,richardson2001capacity,chung2001design,mansour2003high,hocevar2004reduced,zhang2010efficient,hailes2015survey}, they have become the gold standard for high performance, long, large redundancy, SI error correction codes. They are used, for example, in standards including ATSC 3.0 \cite{kim2016low} and 5G New Radio  \cite{richardson2018design}.

A common principle behind the design of long, high-redundancy codes is to create them by concatenating and composing shorter component codes. With \ac{SI}, each component is decoded and provides \ac{SO} that informs the decoding of another component. The process is repeated, passing updated soft information around the code until a global consensus is found or the effort is abandoned. Core to performance is a combination of code-structure, which determines how soft information is circulated amongst components, and the quality of the decoder's \ac{SO}, which captures how much is gleaned from each decoding of a component code. E.g., turbo decoding of Elias's product codes \cite{elias_error-free_1954} can avail of powerful component codes, such as \ac{BCH} codes, but relies on approximate \ac{SO} \cite{pyndiah_1998}, while \ac{LDPC} codes avail of weak single parity-check codes but provide high quality per-component \ac{SO}. 

\Ac{GRAND} is a recently developed family of code-agnostic decoding algorithms that can accurately decode any moderate redundancy code in both hard \citep{duffy_capacity-achieving_2019,galligan2021,An22} and soft detection \citep{Duffy19a,duffy2022_ordered,abbas2021list,an2023soft,Duffy23ORBGRANDAI} settings. \Ac{GRAND} algorithms function by sequentially inverting putative noise effects, ordered from most to least likely according to channel properties and soft information, from received signals. The first codeword yielded by inversion of a noise effect is a maximum-likelihood decoding. Since this procedure does not depend on codebook structure, \ac{GRAND} can decode any moderate redundancy code, even non-linear ones \cite{cohen2023aes}. Efficient hardware implementations \cite{Riaz21,Riaz23,Burg24} and syntheses \cite{condo2021_highperformance, abbas2020grand, abbas2021orbgrand, condo2021fixed} for both hard and  soft-detection settings have established the flexibility and energy efficiency of \ac{GRAND} decoding strategies.

Here we significantly extend \ac{GRAND}'s range of operation by developing \ac{SO} variants that can be used to accurately and efficiently decode long, high-redundancy codes. While initial attempts have used \ac{GRAND} solely as a list decoder \cite{condo2022_iterative, galligan2023_block, Hadavian23} with Pyndiah's traditional approach to generate \ac{SO} \cite{pyndiah_1998}, central to the enhanced performance here is proving that any \ac{SI} \ac{GRAND} algorithm can itself readily produce more accurate \ac{SO}. This enhancement results in better \ac{BLER} and \ac{BER} performance as well as lower complexity. 

The \ac{SO} measure we develop is  an accurate estimate of the \ac{APP} that a decoding is correct in the case of a single decoding, and, in the case of list decoding, the probability that each codeword in the list is correct or the correct codeword is not in the list. We derive these probabilities for uniform at random codebooks and demonstrate that the resulting formulae continue to provide accurate SO for structured codebooks. The formulae can be used with any algorithm in the \ac{GRAND} family regardless of the algorithm's query order. 

Calculating the \ac{SO} only requires knowledge of the code's dimensions and the evaluation of the probability of each noise effect query during GRAND's normal operation, so computation of the measure does not significantly increase the decoder's algorithmic complexity or memory requirements. 

In Forney's seminal work \cite{forney1968_exponential},
the classical per-block \ac{SO} approximation from list decoding is the conditional probability that an element of the list is the decoding given the transmitted codeword is in the list.
The core distinction of \ac{SOGRAND} is that it foregoes the conditioning by inherently including an estimate of the likelihood that the correct decoding has not been found. Its use significantly improves the quality of the information that circulates after decoding a component code and, in contrast to LDPC codes, allows the practical use of powerful component codes with multiple bits of redundancy.

\begin{figure}
	\centering
	\footnotesize
	\begin{tikzpicture}[scale=1]
\footnotesize
\begin{semilogyaxis}[
legend style={at={(0,0)},anchor= south west, font=\scriptsize},
ymin=1e-7,
ymax=1,
width=3.5in,
height=3in,
grid=both,
xmin = 0,
xmax = 2.5,
xlabel = $E_b/N_0$ in dB,
ylabel = {BLER},
]

\addplot[black,mark=o]
table[]{x y
0 0.735294117647059
0.25 0.666666666666667
0.5 0.442477876106195
0.75 0.210084033613445
1.0 0.0740740740740741
1.25 0.0190766882869134
1.5 0.00366891693572058
1.75 0.000641099614058032
2 7.76624231918635e-05
2.25 1.05968192163567e-05
2.5 1.76212385263711e-06

};\addlegendentry{LDPC, BP}

\addplot[cyan,mark=o]
table[]{x y
0 0.95238
0.25 0.95238
0.5 0.55556
0.75 0.38462
1 0.19417
1.25 0.054645
1.5 0.016142
1.75 0.00225
2 0.00024203
2.25 3.6026e-05
2.5 5.9726e-06
};\addlegendentry{LDPC, norm-min-sum}

\addplot[blue,mark=o]
table[]{x y
0 1
0.25 0.86957
0.5 0.64516
0.75 0.4878
1 0.18182
1.25 0.062827
1.5 0.016694
1.75 0.0016926
2 1.2086e-04
2.25 1.0122e-05
};\addlegendentry{dRM prod., ORBGRAND Pyndiah}

\addplot[red,mark=x]
table[]{x y
0 0.96234
0.25 0.87946
0.5 0.65605
0.75 0.22635
1.0 0.07194
1.25 0.01776
1.5 0.001743
1.75 0.0001904
2 1.7705e-05
2.25 1.5675e-06
};\addlegendentry{dRM prod., SOGRAND}

\end{semilogyaxis}

\end{tikzpicture}
	\caption{\ac{AWGN} \ac{BLER} performance of the $(1024, 441)$ 5G LDPC with maximum iteration number $50$ as compared to a $(1024, 441)=(32,21)^2$ dRM product code decoded two ways. First, turbo-decoded \cite{pyndiah_1998}, with $\alpha$ and $\beta$ parameters taken from there and maximum iteration number $20$, but using $1$line-ORBGRAND for list decoding with list size $L=4$ as in \cite{galligan2023_block}. Second, turbo-decoded using $\alpha=0.5$ and \ac{SOGRAND} with max iteration $20$, where lists are added to until $L=4$ or the predicted list-BLER is below $10^{-5}$.}
	\label{fig:prod_grand_32_21}
\end{figure}
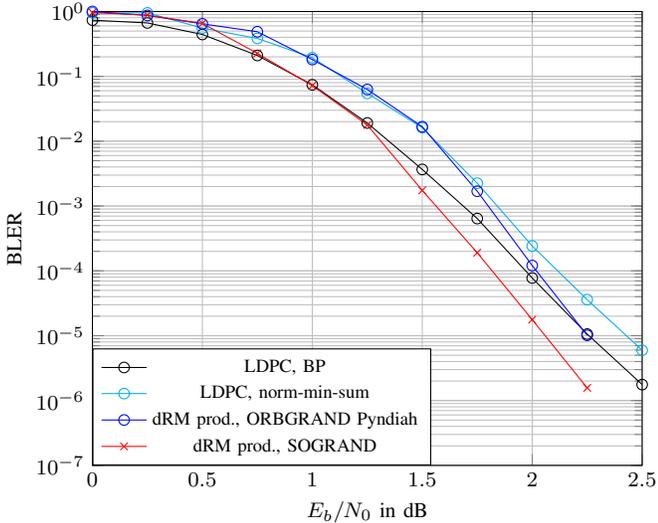

The merit of \ac{SOGRAND} is demonstrated in Fig.~\ref{fig:prod_grand_32_21} where \ac{BLER} vs Eb/N0 in dB for the $(1024,441)$ LDPC in the 5G New Radio standard decoded with \ac{BP} or min-sum \cite{richardson2008modern}
is plotted in the standard \ac{AWGN} setting. Also shown is the performance of a simple product code \cite{elias_error-free_1954}, a $(1024,441)=(32,21)^2$ \ac{dRM} code~\cite{CP21}, decoded two distinct ways. Using Pyndiah's original SO methodology \cite{pyndiah_1998,galligan2023_block} but with a list provided by the 1-line version of \ac{ORBGRAND} \cite{duffy2022_ordered} that enhances its performance over traditional Chase decoding, the product code mildly under-performs in comparison to the \ac{LDPC} code. Decoding with the new \ac{SOGRAND} algorithm developed here sees the product code get a ca. 0.2~dB gain in performance and outperforming the LDPC code. It will be shown in Sec.~\ref{sec:peva} that this improved performance comes with reduced complexity, consistent with previous findings for \ac{GRAND} algorithms, and that the gain is sustained in channels subject to fading.

Comparison of \ac{LDPC} codes and product codes of other dimensions are shown in Sec.~\ref{sec:peva} and establish consistent or better findings: when decoded with \ac{SOGRAND}, Elias's product codes provide better performance than state-of-the-art codes and do so with modest complexity and minimal latency. Variants on the theme of product codes, such as staircase codes \cite{smith2012staircase} and \ac{GLDPC} codes \cite{Liva08}, were developed to enhance the minimum distance of the original design. In section \ref{sec:peva} we show that decoding the \ac{GLDPC} code proposed in \cite{Lentmaier10} with \ac{SOGRAND} results in significantly steeper \ac{BLER} curves and better performance than \ac{LDPC} codes in the 5G standard at higher SNR.

\section{Per-block SO}
For any channel coding scheme it would be desirable if error correction decoders could produce SO in the form of a confidence measure in the correctness of a decoded block. In seminal work on error exponents, Forney \cite{forney1968_exponential} proposed an approximate computation that necessitates the use of a list decoder. It takes the form of a conditional likelihood of correctness given the codeword is in the list. We shall demonstrate that the approach provides an inaccurate estimate in channels with challenging noise conditions. Regardless, its potential utility motivated further investigation, e.g.~\cite{Hof2010}. 

With the recent introduction of \ac{CA-Polar} codes to communications standards \cite{3gpp.38.212}, Forney's approximation has received renewed interest~\cite{sauter2023_error} as one popular method of decoding \ac{CA-Polar} codes, \ac{CA-SCL} decoding, generates a list of candidate codewords as part of its execution, e.g. \cite{niu2012crc,tal2015_list,balatsoukas2015llr}. For convolution or trellis codes, the Viterbi algorithm \cite{forney1973_viterbi} can be modified to produce SO at the sequence level \cite{raghavan1998_reliability}, which has been used in coding schemes with multiple layers of decoding \cite{hagenauer1989_viterbi} and to inform repeat transmission requests \cite{yamamoto1980_viterbi}\cite{raghavan1998_reliability}. 

\ac{SOGRAND} can be used with any moderate redundancy component code, which opens up a large palette of possibilities. Unlike Forney-style approaches, it includes a term for the likelihood that the correct decoding has not yet been found. As a result, it can be evaluated without the need to list decode and the estimate remains accurate in noisy channel conditions. From it, accurate per-bit \ac{SO} can be created through a weighted mixture of identified decodings and \ac{SI} with the likelihood that the decoding is not an element of the list.

\section{Preliminaries}\label{sec:GRAND}
\subsection{Notations}

Let $\codebook$ be a codebook containing $2^k$ binary codewords each of length $n$ bits. Let $\rvword$ be a codeword drawn uniformly at random from the codebook and let $\rvnoise$ denote the binary noise effect that the channel has on that codeword during transmission; that is, $N^n$ encodes the binary difference between the demodulated received sequence and the transmitted codeword. Then $\rvdemodout = \rvword \oplus \rvnoise$ is the demodulated channel output, with $\oplus$ being the element-wise binary addition operator. Let $\rvchanout$ denote real soft channel output per-bit.
Lowercase letters represent realizations of \acp{RV}, with the exception of $\noise$, which is the realization of $\rvnoise$, which is assumed independent of the channel input. We denote the memoryless channel by $f_{R|X}$ and represent both the \ac{PDF} and \ac{PMF} of the \ac{RV} $X$ as $p_{X}(x)$.

\subsection{Background on GRAND}
All \ac{GRAND} algorithms operate by progressing through a series of noise effect guesses $\noiseguess{1}, \noiseguess{2}, \ldots \in \binaryAlphabetN$, whose order is informed by channel statistics, e.g. \cite{An22}, or SI, e.g. \cite{duffy2022_ordered}, until it finds one, $\noiseguess{q}$, that satisfies $\cleaned_q = \demodout \ominus \noiseguess{q} \in \codebook$, where $\ominus$ inverts the effect of the noise on the channel output. If the guesses are in decreasing order of likelihood, then $\noiseguess{q}$ is a maximum-likelihood estimate of $\rvnoise$ and $\cleaned_q$ is a maximum-likelihood estimate of the transmitted codeword $\rvword$. 
Since this guessing procedure does not depend on codebook structure, \ac{GRAND} can decode any moderate redundancy code as long as it has a method for checking codebook membership. For a linear block code with an $(n-k)\times n$ parity-check matrix $\mathbf{H}$, $\cleaned_q$ is a codeword if $\mathbf{H} \cleaned_q = 0^n$ \cite{lin_error_2004}, where $\cleaned_q$ is taken to be a column vector and $0^n$ is the zero vector. To generate a decoding list $\decodelist$ of size $L$, \ac{GRAND}'s guessing procedure continues until $L$ codewords are found \cite{abbas2021list,galligan2023_block}.

While GRAND algorithms can decode any block code, so-called ``even'' linear block codes, e.g. \cite{lin_error_2004}, offer a complexity advantage at no cost in decoding performance so long as the GRAND query generator can skip patterns of a certain Hamming weight \cite{rowshan2022_constrained}. A binary linear code with generator matrix $\mathbf{G} =\{\mathbf{G}_{i,j}\}$ is even if, with addition in $\Fb_2$, $\sum_{j=1}^n \mathbf{G}_{i,j}=0 $ for all $i\in\{1,\ldots,k\}$ as that ensures that $\sum_{i=1}^n c_i = 0$ for every $c^n\in\codebook$.
Consequently, if a codeword, $c^n$, is perturbed by a channel, $y^n=c^n+z^n$, then
\begin{align*}
    \sum_{i=1}^n y_i = \sum_{i=1}^n (c_i+z_i)  =  \sum_{i=1}^n c_i + \sum_{i=1}^n z_i  = \sum_{i=1}^n z_i 
\end{align*}
and the parity of the demodulated bits $y^n$ matches the parity of the noise-effect $z^n$ that has impacted the codeword $c^n$. 

Thus, if an even code is decoded with a GRAND algorithm that has a mechanism to generate noise-effect sequences of given Hamming weight, only noise effect sequences whose parity is the same as $\sum_{i=1}^n y_i$ need be generated. The use of this property can reduce the total query number by a factor of up to $2$, but the performance is the same as if sequences of any Hamming weight were generated as sequences of mismatched Hamming weight cannot lead to the identification of a codebook element. For soft detection ORBGRAND algorithms implemented with the landslide algorithm \cite{duffy2022_ordered,Riaz23}, it is trivial to skip given Hamming weights and so this approach is employed here whenever even codes are used. Note that ORBGRAND has been proven to approach capacity \cite{Liuetal23}.

Many codes are inherently even, including \ac{eBCH} codes, polar codes \cite{arikan2009} and \ac{dRM} codes~\cite{CP21} codes. In general, given a generator, $\mathbf{G}$, for a code that is not even, it is straight-forward to extended it by setting $\mathbf{G}_{i,n+1} = \sum_{j=1}^n \mathbf{G}_{i,j}$ in $\Fb_2$ for each $i\in\{1,\ldots,k\}$, which adds one bit to the code and results in an even code. 

Underlying \ac{GRAND} is a race between two \acp{RV}, the number of guesses until the true codeword is identified and the number of guesses until an incorrect codeword is identified. Whichever of these processes finishes first determines whether the decoding identified by \ac{GRAND} is correct. In the \ac{SI} setting, the guesswork function $\guessworkFun : \binaryAlphabetN \to \{1,\ldots, 2^n\}$ depends on SI, $\rvchanout$, and maps a noise effect sequence to its position in GRAND's guessing order, so that $\guessworkFun(\noiseguess{i})=i$. Thus $\guessworkFun(\rvnoise)$ is a \ac{RV} that encodes the number of guesses until the transmitted codeword would be identified. As $N^n$ is independent of $X^n$, we have that
\begin{align*}
 \P(G(N^n)=q|R^n=r^n) &=p_{N^n|\rvchanout}(z^{n,q}|r^n).
\end{align*}
Consequently, as $Y^n=X^n\oplus N^n$, for any $\codeword\in\codebook$  
\begin{align*}
    p_{\rvword|\rvchanout}(\codeword|\chanout) 
    & = P\left(N^n=\codeword \oplus y^n\middle| \rvchanout=\chanout, N^n\oplus y^n\in\codebook\right)
\end{align*}
If $\W{i} : \Omega \to \{1,\ldots, 2^n-1\}$ is the number of guesses until the $i$-th incorrect codeword would be identified, not accounting for the query that identifies the correct codeword, then \ac{GRAND} returns a correct decoding whenever $\guesswork \leq \W{1}$ and a list of length $L$ containing the correct codeword whenever $\guesswork \leq \W{L}$. Stochastic analysis of the race between these two processes leads to the derivation of the SO in this paper.

\subsection{Background on SO}
Given channel output $\chanout$, which serves as SI to the decoder, and a decoding output $\mlcodeword \in \codebook$, the probability that the decoding is correct is
\begin{align}
p_{\rvword|\rvchanout}(\mlcodeword|\chanout)
&=\frac{f_{\rvchanout|\rvword}(\chanout|\mlcodeword)}{\displaystyle\sum_{\codeword\in\codebook} f_{\rvchanout|\rvword}(\chanout|\codeword)}.\label{formula:forney}
\end{align}
Based on this formula, in Forney's work on error exponents \cite{forney1968_exponential} he derived an optimal threshold for determining whether a decoding should be marked as an erasure. Evaluating the sum in the denominator requires $2^k$ computations, which is infeasible for codebooks of practical size. Thus, given a decoding list $\decodelist \subseteq \codebook$ with $L\geq 2$ codewords, Forney suggested the approximation 
\begin{align}
&p_{\rvword|\rvchanout}(\mlcodeword|\chanout) \nonumber\\
&\approx \frac{f_{\rvchanout|\rvword}(\chanout|\mlcodeword)}{\displaystyle\sum_{\codeword\in\decodelist} f_{\rvchanout|\rvword}(\chanout|\codeword)
+ \underbrace{\sum_{\codeword\in\codebook\setminus\decodelist}f_{\rvchanout|\rvword}(\chanout|\codeword)}_{\text{assumed to be } 0}}, \label{eq:missingterm}
\end{align}
which is the conditional probability that each element of the list is correct given that
the transmitted codeword is in the list. With 
$
\mlcodeword = \arg\max_{\codeword\in\decodelist}  f_{\rvchanout|\rvword}(\chanout|\codeword),
$
Forney's approximation results in an estimate of the correctness probability of $\mlcodeword$ that is no smaller than $1/L$. Having the codewords of highest likelihood in the decoding list gives the most accurate approximation as their likelihoods dominate the sum.

A significant source of improvement in the \ac{SO} that results from the method developed here for \ac{GRAND} is that it provides an accurate, non-zero estimate of final term in the denominator of eq. \eqref{eq:missingterm}, even with a single decoding, i.e. $L=1$. As a result, it generates an estimate that is not conditioned on the codeword having been found. That enables an estimate of the likelihood of correctness of a single decoding, i.e. for $L=1$, which is not possible from Forney's approximation, as well as an
estimate of the likelihood that the decoding has not been found (i.e. is not in the list):
\begin{align}
    p_{\rvword|\rvchanout}\left(\codebook{\setminus}\decodelist|\chanout\right)
    = \frac{\displaystyle
    \sum_{\codeword\in\codebook{\setminus}\decodelist} f_{\rvchanout|\rvword}(\chanout|\codeword)}
    {\displaystyle\sum_{\codeword\in\codebook} f_{\rvchanout|\rvword}(\chanout|\codeword)}.\label{eq:MDR}
\end{align}
It is this distinctive feature that enables enable substantially improved \ac{SO}. 

One downside of Forney's approach is that it requires a list of codewords, which most decoders do not provide. For this reason, a method has recently been proposed to estimate the likelihood of the second-most likely codeword given the first \cite{freudenberger2021_reduced}. A variety of schemes have also been suggested for making erasure decisions, e. \cite{hashimoto1999_composite}.

\section{GRAND Per Block SO}
Throughout this section, we shall assume that the codebook, $\codebook$, consists of $2^k$ codewords drawn uniformly at random from $\binaryAlphabetN$, although the derivation generalises to higher-order symbols. We first derive exact expressions, followed by readily computable approximations, for the probability that the transmitted codeword is each element of a decoding list or is not contained within it. As a corollary, we obtain a formula for the probability that a single-codeword \ac{GRAND} output is incorrect. In Section \ref{sec:results} we demonstrate the formulae provide excellent estimates for structured codebooks. 

\begin{theorem}[GRAND list decoding \ac{APP}s for a uniformly random codebook]
\label{thm:list}
Given the soft information $R^n=r^n$ that determines the query order, 
let $G(N^n)$ be the number of codebook queries until the noise effect sequence $N^n$ is identified.
Let $W_1,\ldots,W_{2^k-1}$ be selected uniformly at random without replacement from $\{1,\ldots,2^n-1\}$ and 
define their rank-ordered version $\W{1}<\cdots<\W{2^k-1}$. With the true noise effect not counted, $\W{i}$ corresponds to the location in the guesswork order of the $i$-th erroneous decoding in a codebook constructed uniformly-at-random. 
Define the partial vectors $\Wij{i}{j} = (\W{i},\ldots,\W{j})$ for each $i\leq j\in\{1,\ldots,2^k-1\}$.

Assume that a list of $L\geq 1$ codewords are identified by a GRAND decoder at query numbers $q_1<\ldots<q_L$.\footnote{$q_i$ denotes the query number of the $i$-th found codeword. Note that the $i$-th found codeword is not necessarily the $i$-th most likely codeword.} Define the associated partial vectors $\qij{i}{j} = (q_i,\ldots,q_{j})$ for each $i\leq j\in\{1,\ldots,2^k-1\}$, and 
\begin{align}
\qiLj{i} = (q_1,\ldots,q_{i-1},q_{i+1}-1,\ldots,q_L-1),
\label{eq:qomit}
\end{align}
which is the vector $\qij{1}{L}$ but with the entry $q_i$ omitted and one subtracted for all entries from $q_{i+1}$ onwards.
Define
\begin{align*}
\PA = \sum_{q>q_L} p_{N^n|\rvchanout}(z^{n,q}|r^n) p_{\Wij{1}{L}}(\qij{1}{L}) 
\end{align*}
which is associated with the transmitted codeword not being in the list,
and, for each $i\in\{1,\ldots,L-1\}$,
\begin{align*}
\PB{i} &= p_{N^n|\rvchanout}(z^{n,q_i}| r^n) p_{\Wij{1}{L-1}}(\qiLj{i}), 
\end{align*}
which is associated with the transmitted codeword being the $i$-th element of
the list,
and
\begin{align*}
\PB{L}  &=  p_{N^n|\rvchanout}(z^{n,q_L}|\chanout)
    \sum_{q\geq q_L} p_{\Wij{1}{L-1},\W{L}}(\qij{1}{L-1},q),
\end{align*}
which is associated with the transmitted codeword being the final element of the
list. Then the probability that $i$-th codeword is correct is
\begin{align}
p_{\rvword|\rvchanout}\left(y^n\oplus z^{n,q_i}|\chanout\right) =
\frac{\PB{i}}{\sum_{i=1}^L \PB{i} + \PA}
\label{eq:app_correct}
\end{align}
and the probability that the correct decoding 
is not in the list is
\begin{align}
p_{\rvword|\rvchanout}\left(\codebook{\setminus}\decodelist|\chanout\right) =
\frac{\PA}{\sum_{i=1}^L \PB{i} + \PA}.
\label{eq:app}
\end{align}

\end{theorem}

\begin{proof}
For $q\in\{1,\ldots,2^n\}$, define $\Wq{i}{q} = \W{i} + 1_{\{\W{i}\geq q\}}$,
so that any $\W{i}$ that is greater than or equal to $q$ is incremented by one. Note that $\Wq{i}{G(N^n)}$ encodes the locations of erroneous codewords in the guesswork order of a randomly constructed codebook given the value of $G(N^n)$ and, in particular, $\Wq{i}{G(N^n)}$ corresponds the number of queries until the $i$-th incorrect codeword is found given $G(N^n)$.

Given the soft information $R^n=r^n$ that determines the query order, we identify the event that the decoding is not in the list as
\begin{align*}
\A = \left\{G(N^n)>q_L,\Wij{1}{L}=\qij{1}{L}\right\}
\end{align*}
and the events where the decoding is the $i$-th element of the list by
\begin{align*}
\B{i} = &\left\{\Wij{1}{i-1}=\qij{1}{i-1}, G(N^n)=q_i, \right. \\
        & \left. \Wij{i}{L-1}+1=\qij{i+1}{L}, \W{L}\geq q_{L} \right\}
\end{align*}
where the final condition is automatically met for $i=\{1,\ldots,L-1\}$ but not for $i=L$.
The conditional probability that a GRAND decoding is not one of the elements in the list given that $L$ elements have been found is
\begin{align}
&\P\left(\A\middle|\A\bigcup_{i=1}^L\B{i},R^n=r^n\right)\nonumber\\
&= \left. \P(\A|R^n=r^n) \middle/ \P\left(\A\bigcup_{i=1}^L\B{i}\right|R^n=r^n\right).
\label{eq:LLR}
\end{align}
As the $\A$ and $\B{i}$ events are disjoint, 
to compute eq. \eqref{eq:LLR} it suffices to simplify $\P(\A)$ and $\P(\B{i})$ for $i\in\{1,\ldots,L\}$ to
evaluate the \ac{APP}s.

Consider the numerator,
\begin{align*}
    \P\left(\A\right) 
    &= \P(G(N^n)>q_L, \Wij{1}{L}=\qij{1}{L}|R^n=r^n)\\
    &= \sum_{q>q_L} p_{N^n|\rvchanout}(z^{n,q}|r^n) p_{\Wij{1}{L}}(\qij{1}{L}),
\end{align*}
where we have used the fact that $G(N^n)$ is independent of $\Wij{1}{L}$ by construction. In considering the denominator, we need only be concerned with the terms $\P(\B{i})$
corresponding to a correct codebook being identified at query $q_i$, for which 
\begin{align*}
\P(\B{i})
= &\P(G(N^n)=q_i,\Wij{1}{i-1}=\qij{1}{i-1}, \\
  & \Wij{i}{L-1}+1=\qij{i+1}{L}, \W{L}\geq q_{L} |R^n=r^n)\\
= & P(G(N^n)=q_i, \Wij{1}{L-1}=\qiLj{i}, \W{L}\geq q_{L} |R^n=r^n)\\
= & p_{N^n|\rvchanout}(z^{n,q_i}| r^n) 
  \sum_{q\geq q_L} p_{\Wij{1}{L-1},\W{L}}(\qiLj{i},q),
\end{align*}
where we have used the definition of $\qiLj{i}$ in eq. \eqref{eq:qomit} and independence. Thus the conditional probability that the correct answer is not found in eq. \eqref{eq:LLR} is given in eq. \eqref{eq:app} and the conditional probability that a given element of the list is correct is given by eq. \eqref{eq:app_correct}.
\end{proof}

As a result of Theorem \ref{thm:list}, to compute the list decoding \ac{APP}s one needs to evaluate or approximate: 
\begin{align*}
    \text{1) } & p_{N^n|\rvchanout}( z^{n,q}|\chanout) \text{ and } \sum_{j\leq q} p_{N^n|\rvchanout}( z^{n,j}|\chanout);\\
    \text{2) } & p_{\Wij{1}{L}}(\qij{1}{L}) \text{ and } \sum_{q\geq q_L} p_{\Wij{1}{L-1},\W{L}}(\qij{1}{L-1},q).
\end{align*}
During a GRAND algorithm's execution, the evaluation of 1) can be achieved by calculating the likelihood of each noise effect query as it is made and retaining a running sum. For 2), geometric approximations, whose asymptotic precision can be verified using the approach described in \cite[Theorem 2]{duffy_capacity-achieving_2019}, can be employed, result in the following corollaries for list decoding and single-codeword decoding, respectively. Note that these estimates only need to know the code dimensions, $n$ and $k$, rather than any specifics of the code construction.

\begin{corollary}[Approximate \ac{APP}s for a random codebook]
If each $\W{i}$ given $\W{i-1}$ is assumed to be geometrically distributed with probability of success $(2^k-1)/(2^n-1)$, 
then 
eq. \eqref{eq:app_correct} describing the \ac{APP} that $i$-th found decoding is correct, $p_{\rvword|\rvchanout}\left(y^n\oplus z^{n,q_i}|\chanout\right)$ can be approximated as
\begin{align}
\frac{\displaystyle p_{N^n|\rvchanout}(z^{n,q_i}|\chanout) }
{
\begin{array}{l}
\displaystyle 
\sum_{i=1}^Lp_{N^n|\rvchanout}(z^{n,q_i}|\chanout) \\
\displaystyle \qquad + \left(1-\sum_{j=1}^{q_L} p_{N^n|\rvchanout}(z^{n,j}|\chanout)\right) \left(\frac{2^k-1}{2^n-1}\right) 
\end{array}
}
\label{eq:app_approx_correct}
\end{align}
and the probability that the list does not contain the transmitted codeword, $p_{\rvword|\rvchanout}\left(\codebook{\setminus}\decodelist|\chanout\right)$ approximated by
\begin{align}
\frac{\displaystyle \left(1-\sum_{j=1}^{q_L} p_{N^n|\rvchanout}(z^{n,j}|\chanout)\right) \left(\frac{2^k-1}{2^n-1}\right) }
{
\begin{array}{l}
\displaystyle 
\sum_{i=1}^Lp_{N^n|\rvchanout}(z^{n,q_i}|\chanout) \\
\displaystyle \qquad + \left(1-\sum_{j=1}^{q_L} p_{N^n|\rvchanout}(z^{n,j}|\chanout)\right) \left(\frac{2^k-1}{2^n-1}\right) 
\end{array}
}
\label{eq:app_approxlist}
\end{align}

\end{corollary}
\begin{proof}
Define the geometric distribution's probability of success to be $\PHitIncorrect=(2^k-1)/(2^n-1)$.
Under the assumptions of the corollary, we have the formulae
\begin{align*}
p_{\Wij{1}{L}}(\qij{1}{L}) &= \left(1-\PHitIncorrect\right)^{q_L-L} \PHitIncorrect^L,
\end{align*}
for $i\in\{1,\ldots,L-1\}$
\begin{align*}
p_{\Wij{1}{L-1}}(\qiLj{i}) & = \left(1-\PHitIncorrect\right)^{q_L-L} \PHitIncorrect^{L-1},
\end{align*}
and
\begin{align*}
\P\left(\Wij{1}{L-1}=\qij{1}{L-1},\W{L}\geq q_L\right)
& = \left(1-\PHitIncorrect\right)^{q_L-L} \PHitIncorrect^{L-1}.
\end{align*}
Using those expressions, simplifying eq. \eqref{eq:app_correct} gives eq. \eqref{eq:app_approx_correct} and
eq. \eqref{eq:app} gives eq.~\eqref{eq:app_approxlist}.
\end{proof}

Taken together, this theorem and corollary provide a simple methodology by which \ac{GRAND} can provide an \ac{APP} of the correctness of each element of a decoding list, as well as the likelihood that the correct decoding has not been found. 

The proposed SOGRAND algorithm is compatible with any GRAND algorithm. In addition to the usual operations required for GRAND, we require the following additional computations: during the initialization phase, the probability of all-zero noise, i.e., $p_{N^n\left|R^n\right.}\left(0^n\left|r^n\right.\right)$, is computed. Each time GRAND generates a test pattern $z^n$, we need to compute the probability of this pattern $p_{N^n\left|R^n\right.}\left(z^n\left|r^n\right.\right)$ and accumulate it accordingly. The term $p_{N^n\left|R^n\right.}\left(z^n\left|r^n\right.\right)$ can be computed based on $p_{N^n\left|R^n\right.}\left(0^n\left|r^n\right.\right)$, which needs $w(z^n)$ operations, where $w(z^n)$ denotes the Hamming weight of the test noise pattern $z^n$, which is usually sparse. In total, we need $q+\sum_{j=1}^{q} w\left(z^{n,j}\right)$ additional operations and store additional values, such as the probability of all-zero noise, the probability of the current guess and the accumulation of these guesses.

\section{SO accuracy}
\label{sec:results}
\subsection{Blockwise SO accuracy}
Armed with the approximate \ac{APP} formulae, we investigate its precision for random and structured codebooks. Transmissions were simulated using an \ac{AWGN} channel with \ac{BPSK} modulation. \ac{ORBGRAND} \cite{duffy2022_ordered} was used for \ac{SI} decoding, which produced decoding lists of the appropriate size for both \ac{SO} methods. 
\begin{figure}
\centering
	\footnotesize
	\begin{tikzpicture}[scale=1]
\footnotesize
\begin{loglogaxis}[
legend style={at={(1,0)},anchor= south east, font=\scriptsize},
ymin=1e-5,
ymax=1,
width=3.5in,
height=3.5in,
grid=both,
xmin = 1e-5,
xmax = 1,
xlabel = {predicted},
ylabel = {empirical},
]

\addplot[blue,mark = o]
table[]{x y
0.55261 0.54738
0.2082 0.1939
0.063663 0.055376
0.020279 0.017503
0.0065341 0.0060053
0.0021236 0.0020131
};\addlegendentry{$(64,57)$ eBCH, SOGRAND}

\addplot[red,mark=o]
table[]{x y
0.4338 0.43954
0.18976 0.17771
0.060896 0.050814
0.019147 0.017472
0.0060856 0.0058309
0.0019805 0.0016864
0.00064631 0.00058232
};\addlegendentry{$(32,26)$ eBCH, SOGRAND}


\addplot[brown,mark=+]
table[]{x y
0.41973 0.4403
0.18312 0.16942
0.058928 0.049497
0.018728 0.014667
0.0059058 0.0047478
0.0018682 0.0014175
0.00059906 0.00058774
0.00019111 0.00018206
};\addlegendentry{$(32,21)$ dRM, SOGRAND}

\addplot[cyan,mark=o]
table[]{x y
0.39805 0.4
0.17179 0.16204
0.058354 0.055239
0.018665 0.016643
0.0059389 0.0056039
0.0018812 0.0016671
0.00059974 0.00053002
0.00019266 0.00014562
};\addlegendentry{$(16,11)$ eBCH, SOGRAND}

\addplot[blue,mark = o,dashed,mark options=solid]
table[]{x y
0.061021 0.34496
0.029863 0.26394
0.0099756 0.15061
0.0031675 0.067406
0.00099553 0.030114
0.0003172 0.011404
0.00010063 0.0052469
3.298e-05 0.0038941
};\addlegendentry{$(64,57)$ eBCH, Forney}

\addplot[red,mark=o,dashed,mark options=solid]
table[]{x y
0.060195 0.22836
0.02752 0.15837
0.0093416 0.083176
0.0029963 0.031421
0.00095847 0.012555
0.00030466 0.0042706
9.6353e-05 0.0013043
3.0864e-05 0.00059112
};\addlegendentry{$(32,26)$ eBCH, Forney}


\addplot[brown,mark=+,dashed,mark options=solid]
table[]{x y
0.058908 0.20435
0.026758 0.13078
0.0091157 0.065214
0.0029089 0.027236
0.00092845 0.011601
0.00029659 0.0032741
9.3565e-05 0.0011778
2.9924e-05 0.00071429
};\addlegendentry{$(32,21)$ dRM, Forney}

\addplot[cyan,mark=o,dashed,mark options=solid]
table[]{x y
0.060028 0.14509
0.026106 0.097052
0.0089243 0.043009
0.0029072 0.022639
0.00092802 0.0093391
0.00029718 0.0043343
9.4346e-05 0.0013612
3.0173e-05 0.00036813
};\addlegendentry{$(16,11)$ eBCH, Forney}

\addplot[gray,dashed]
table[]{x y
1e-5 1e-5
1 1
};

\end{loglogaxis}

\end{tikzpicture}
	\caption{Predicted list-BLER based on soft-output vs. empirical list-BLER ($L^\prime=4$): \ac{SOGRAND} with $L=4$, Forney with $L=5$, $E_b/N_0=2$. }
	\label{fig:list_BLER}
\end{figure}
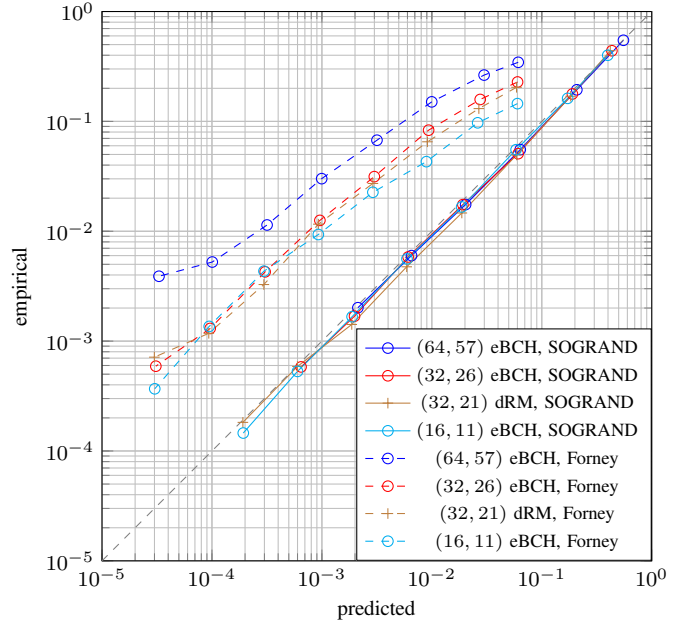
Fig.~\ref{fig:list_BLER} plots the empirical list-\ac{BLER} given the predicted list-\ac{BLER} evaluated using eq. \eqref{eq:app_approxlist}. A list error occurs when the transmitted codeword is not in the list of size $L^\prime$. If the estimate was precise, then the plot would follow the line $x=y$, as the predicted list-\ac{BLER} and the list-\ac{BLER} would match. It can be seen that \ac{SOGRAND} provides near perfect estimates for codes of different structures and dimensions.

\subsection{Bitwise SO accuracy}\label{sec:siso}
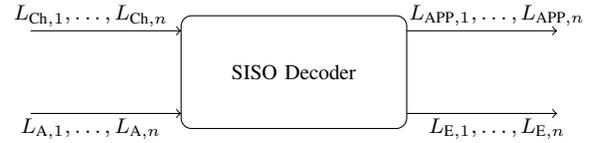
\begin{figure}[t]
	\centering
	\begin{tikzpicture}
\footnotesize
\draw[->] (-2,0.2) to (0,0.2);
\node at (-1.2,0) {$L_{\text{A},1},\dots,L_{\text{A},n}$};
\draw[->] (-2,1.3) to (0,1.3);
\node at (-1.2,1.5) {$L_{\text{Ch},1},\dots,L_{\text{Ch},n}$}; 

\draw[rounded corners]  (0,0) rectangle (3,1.5);
\node at (1.5,0.75)[align=center] {SISO Decoder};

\draw[->] (3,0.2) to (5,0.2);
\draw[->] (3,1.3) to (5,1.3);
\node at (4.2,0) {$L_{\text{E},1},\dots,L_{\text{E},n}$}; 
\node at (4.2,1.5) {$L_{\text{APP},1},\dots,L_{\text{APP},n}$}; 

\end{tikzpicture}
	\caption{Soft-Input Soft-Output (SISO) decoder schematic.}
	\label{fig:siso}
\end{figure}
For many applications, the system requires post-decoding bitwise SO. Fig.~\ref{fig:siso} shows the inputs and outputs of a \ac{SISO} decoder. A \ac{SISO} decoder takes the sum of channel \acp{LLR} $L_{\text{Ch},i}$ and a-priori \acp{LLR} $L_{\text{A},i}$ as input, where
\begin{align*}
L_{\text{Ch},i} = \log\frac{p_{R|X}(\chanouti|0)}{p_{R|X}(\chanouti|1)},~L_{\text{A},i} = \log\frac{p_{X_i}(0)}{p_{X_i}(1)},~i=1,\dots,n.
\end{align*}
and returns \ac{APP} \acp{LLR} $L_{\text{APP},i}$ and extrinsic \acp{LLR} $L_{\text{E},i}$ as output, where
\begin{align*}
L_{\text{APP},i} &= \log\frac{p_{X_i|\rvchanout}(0|\chanout) }{p_{X_i|\rvchanout}(1|\chanout) },
L_{\text{E},i} = L_{\text{APP},i} - L_{\text{A},i} - L_{\text{Ch},i},
\end{align*}
for $i\in\{1,\ldots,n\}$.
Based on the \ac{APP}s approximations in eq. \eqref{eq:app_approx_correct} and \eqref{eq:app_approxlist}, the bitwise SO \ac{APP} LLRs $L_{\text{APP},i}$ is
\begin{align}
&L^\prime_{\text{APP},i}=\nonumber
\\
&\log 
    \frac{\displaystyle
    \sum_{\codeword\in\decodelist:\codewordi=0} p_{\rvword|\rvchanout}(\codeword|\chanout) 
        + p_{\rvword|\rvchanout}\left(\codebook{\setminus}\decodelist|\chanout\right) p_{X|\rvchanouti}(0|\chanouti)
     }
     {\displaystyle 
     \sum_{c^n\in\decodelist:\codewordi=1} p_{\rvword|\rvchanout}(\codeword|\chanout) + p_{\rvword|\rvchanout}\left(\codebook{\setminus}\decodelist|\chanout\right) p_{X|\rvchanouti}(1|\chanouti)
     } \label{eq:LLRi}
\end{align}
which reflects the content of each codeword in the list and it's likelihood, in addition to the likelihood that the codeword is not found whereupon the prior information is retained.

When compared to Pyndiah's approximation, 
\begin{align*}
\log 
    \frac{\displaystyle
    \max_{\codeword\in\decodelist:\codewordi=0} p_{\rvword|\rvchanout}(\codeword|\chanout)}
     {\displaystyle 
     \max_{c^n\in\decodelist:\codewordi=1} p_{\rvword|\rvchanout}(\codeword|\chanout) 
     }
\end{align*}
eq. \eqref{eq:app_approxlist} introduces an additional term that dynamically adjusts the weight between the decoding observations and channel observation. Furthermore, eq. \eqref{eq:app_approxlist} eliminates the need for the saturation value present in Pyndiah's approximation. Fig.~\ref{fig:BER_prediction} shows the \ac{BER} prediction obtained from eq. \eqref{eq:LLRi} and ORBGRAND decoding. The simulated results show that the new \ac{SOGRAND} approach predicts the \ac{BER} accurately across a range of channel conditions, code rates and types.
\begin{figure}
\centering
	\footnotesize
	\begin{tikzpicture}[scale=1]
\footnotesize
\begin{axis}[
legend style={at={(0,1)},anchor= north west, font=\scriptsize},
ymin=1e-5,
ymax=1,
width=3.5in,
height=3.5in,
grid=both,
xmin = 1e-5,
xmax = 1,
xlabel = {predicted},
ylabel = {empirical},
]

\addplot[blue,mark = o]
table[]{x y
0.0042667 0.0079936
0.072027 0.095383
0.12336 0.15094
0.17385 0.20156
0.22414 0.24977
0.2742 0.29748
0.32432 0.34485
0.37426 0.39003
0.42434 0.43878
0.474 0.48904
0.52382 0.52618
0.57398 0.57346
0.62404 0.61826
0.67435 0.67164
0.724 0.72679
0.77425 0.77788
0.8247 0.82478
0.87423 0.87164
0.9245 0.93253
0.97222 0.97424
};\addlegendentry{$(64,57)$ eBCH, SOGRAND}

\addplot[red,mark=o]
table[]{x y
0.0029083 0.0063255
0.071837 0.092362
0.12315 0.14109
0.17365 0.1945
0.22395 0.24662
0.27417 0.29718
0.32437 0.34667
0.37447 0.39701
0.42455 0.43924
0.47453 0.4859
0.52443 0.53456
0.57476 0.57555
0.62478 0.63004
0.67508 0.6838
0.725 0.72242
0.77502 0.77431
0.82497 0.82561
0.87577 0.87872
0.92638 0.92471
0.97999 0.98028

};\addlegendentry{$(32,26)$ eBCH, SOGRAND}


\addplot[brown,mark=+]
table[]{x y
0.0019082 0.0034406
0.07175 0.086417
0.12319 0.13724
0.17392 0.19214
0.22379 0.25271
0.27412 0.28888
0.32424 0.34281
0.37461 0.39534
0.42467 0.44582
0.47477 0.48529
0.52493 0.54424
0.57534 0.59133
0.62537 0.62356
0.67527 0.69187
0.72522 0.7298
0.77529 0.774
0.82595 0.8326
0.87623 0.88263
0.92722 0.93193
0.99026 0.99077

};\addlegendentry{$(32,21)$ dRM, SOGRAND}

\addplot[cyan,mark=o]
table[]{x y
0.0025708 0.0044987
0.071551 0.085537
0.12294 0.13808
0.17362 0.18932
0.22396 0.24233
0.27414 0.29257
0.32441 0.34212
0.37445 0.39184
0.4246 0.43644
0.47477 0.48819
0.52483 0.53569
0.57503 0.58132
0.62498 0.63375
0.67522 0.68564
0.72532 0.73218
0.77541 0.78053
0.82576 0.82932
0.87623 0.87889
0.9272 0.92845
0.98594 0.98627
};\addlegendentry{$(16,11)$ eBCH, SOGRAND}

\addplot[blue,mark = o,dashed,mark options=solid]
table[]{x y
0.00060594 0.014203
0.073576 0.12274
0.12411 0.16252
0.17439 0.19741
0.22456 0.22049
0.27466 0.24346
0.32457 0.26963
0.3746 0.28813
0.42475 0.30905
0.47484 0.33797
0.52442 0.36231
0.57454 0.39444
0.62476 0.41858
0.67467 0.45058
0.72479 0.48484
0.77485 0.53151
0.82483 0.58375
0.87513 0.6425
0.92532 0.73114
0.97904 0.8613
};\addlegendentry{$(64,57)$ eBCH, Pyndiah}

\addplot[red,mark=o,dashed,mark options=solid]
table[]{x y
0.0014532 0.0080093
0.072381 0.089433
0.12358 0.13406
0.17397 0.17148
0.22421 0.21306
0.27424 0.23471
0.3245 0.27583
0.37459 0.30891
0.42468 0.33601
0.47468 0.37963
0.52468 0.4087
0.57499 0.45088
0.62497 0.47447
0.67486 0.52959
0.72536 0.5898
0.77533 0.62825
0.82525 0.68198
0.87578 0.75132
0.92644 0.82804
0.98308 0.94991
};\addlegendentry{$(32,26)$ eBCH, Pyndiah}


\addplot[brown,mark=+,dashed,mark options=solid]
table[]{x y
0.0012432 0.0040623
0.072132 0.083656
0.12307 0.13545
0.17369 0.16494
0.22441 0.22722
0.27464 0.23667
0.32489 0.30435
0.37468 0.31666
0.42465 0.36385
0.47466 0.39612
0.525 0.44664
0.57511 0.48629
0.6249 0.52448
0.67454 0.56127
0.72513 0.62429
0.77511 0.66179
0.8257 0.73632
0.87613 0.79217
0.92727 0.86883
0.99166 0.98143
};\addlegendentry{$(32,21)$ dRM, Pyndiah}

\addplot[cyan,mark=o,dashed,mark options=solid]
table[]{x y
0.0020561 0.0050009
0.07168 0.081178
0.12311 0.12831
0.17369 0.1674
0.22422 0.21213
0.27431 0.25191
0.32438 0.2923
0.37448 0.3251
0.42473 0.36962
0.4748 0.41731
0.5248 0.44937
0.57489 0.49109
0.62499 0.5428
0.67506 0.58528
0.72524 0.63432
0.7753 0.68525
0.82573 0.74862
0.87626 0.80728
0.92733 0.87733
0.98747 0.97499

};\addlegendentry{$(16,11)$ eBCH, Pyndiah}

\addplot[gray,dashed]
table[]{x y
1e-5 1e-5
1 1
};

\end{axis}

\end{tikzpicture}
	\caption{\ac{SO} predicted vs. empirical BER: $L=4$, $E_b/N_0=2$.}
	\label{fig:BER_prediction}
\end{figure}
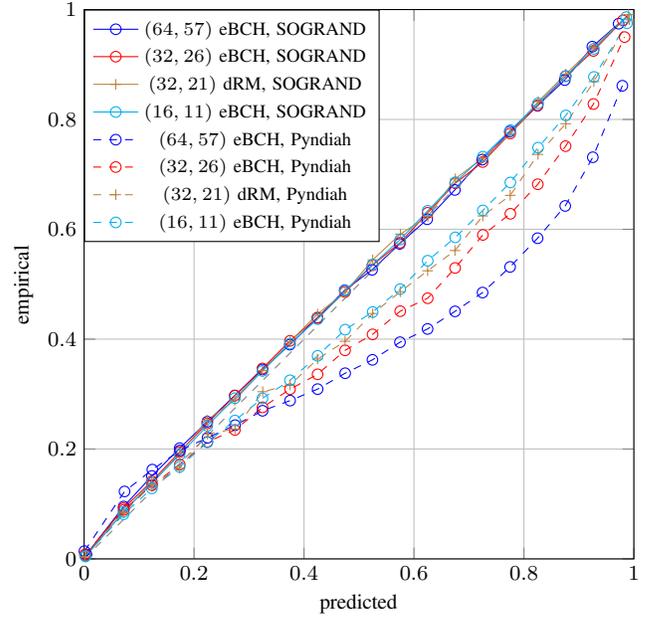

\section{Long Code Performance Evaluation}
\label{sec:peva}

\begin{figure}
	\centering
	\footnotesize
	\begin{tikzpicture}[scale=1]
\footnotesize
\begin{semilogyaxis}[
legend style={at={(0,0)},anchor= south west, font=\scriptsize},
ymin=1e-7,
ymax=1,
width=3.5in,
height=3in,
grid=both,
xmin = 0,
xmax = 4,
xlabel = $E_b/N_0$ in dB,
ylabel = {BLER(solid) / BER(dashed)},
]

\addplot[black,mark=o]
table[]{x y
0 0.88106
0.5 0.65574
1 0.32895
1.5 0.14641
2 0.045249
2.5 0.0075047
3 0.0010019
3.5 8.1631e-05
4 5.794e-06
};\addlegendentry{LDPC, BP}

\addplot[blue,mark=o]
table[]{x y
0 1
0.5 0.81081
1 0.55556
1.5 0.28037
2 0.088235
2.5 0.021292
3 0.0027523
3.5 0.00011041
4 4.8907e-06
};\addlegendentry{polar, SCL} 

\addplot[red,mark=x]
table[]{x y
0 0.85231
0.5 0.55123
1 0.30123
1.5 0.09683
2 0.017001
2.5 0.003542
3 0.0005231
3.5 6.7832e-05
4 7.2112e-06
};\addlegendentry{eBCH prod.}

\addplot[black,mark=o,dashed,mark options=solid]
table[]{x y
0 0.16649
0.5 0.10573
1 0.050756
1.5 0.020873
2 0.0061479
2.5 0.0010122
3 0.00012512
3.5 8.298e-06
4 5.4621e-07
};

\addplot[blue,mark=o,dashed,mark options=solid]
table[]{x y
0 0.50547
0.5 0.40752
1 0.27879
1.5 0.14011
2 0.044324
2.5 0.010718
3 0.0014034
3.5 5.4805e-05
4 2.4683e-06
}; 

\addplot[red,mark=x,dashed,mark options=solid]
table[]{x y
0 0.12423
0.5 0.07123
1 0.03134
1.5 0.0077512
2 0.0016557
2.5 0.00023342
3 3.6422e-05
3.5 4.7023e-06
4 3.9734e-07
};

\end{semilogyaxis}

\end{tikzpicture}

\begin{tikzpicture}[scale=1]
\begin{semilogyaxis}[
legend style={at={(1,1)},anchor= north east, font=\scriptsize},
legend columns=2,
ymin=1e-1,
ymax=1e3,
width=3.5in,
height=1.5in,
grid=both,
xmin = 0,
xmax = 4,
ylabel = {avg. \# of Guesses per bit},
]

\addplot[red,mark=x]
table[]{x y
0 151.29
0.5 110.96
1 52.22
1.5 29.683
2 19.611
2.5 15.407
3 12.729
3.5 10.7931
4 9.2432
};\addlegendentry{eBCH prod.}

\addplot[red,mark=x,dashed,mark options=solid]
table[]{x y
0 17.179
0.5 12.688
1 5.7828
1.5 3.197
2 2.0717
2.5 1.6079
3 1.3215
3.5 1.117
4 0.95903
};\addlegendentry{eBCH prod., parall.} 

\end{semilogyaxis}

\end{tikzpicture}

\begin{tikzpicture}[scale=1]
\begin{semilogyaxis}[
legend style={at={(0,0)},anchor= south west, font=\scriptsize},
legend columns=2,
ymin=1e-1,
ymax=1e2,
width=3.5in,
height=1.5in,
grid=both,
xmin = 0,
xmax = 4,
ylabel = {avg. \# of iterations},
]

\addplot[black,mark=o]
table[]{x y
0 50
0.5 50
1 22.333
1.5 14.907
2 10.634
2.5 7.6697
3 5.9588
3.5 5.0815
4 4.436
};\addlegendentry{LDPC, BP}


\addplot[red,mark=x]
table[]{x y
0 18.021
0.5 12.711
1 5.7
1.5 3.0918
2 1.9763
2.5 1.5267
3 1.2492
3.5 1.0611
4 0.91923
};\addlegendentry{eBCH prod.} 

\end{semilogyaxis}

\end{tikzpicture}
	\caption{AWGN performance of $(256, 121)$ 5G LDPC with max. iterations $50$ and the $(256, 121)$ 5G \ac{CA-Polar} ($24$-bits CRC) decoded with \ac{CA-SCL} ($L=16$) as compared to a $(256, 121)=(16,11)^2$ eBCH product code decoded with \ac{SOGRAND}, with $\alpha=0.5$ and maximum iteration number $20$, where lists are added to until $L=4$ or the predicted list-BLER is below $10^{-5}$. Upper panel: \ac{BLER} and \ac{BER}. Middle panel: average number of queries per-bit until a decoding, where parallelized assumes all rows/columns are decoded in parallel. Lower panel: average number of iterations.}
	\label{fig:prod_grand_16_11}
\end{figure}
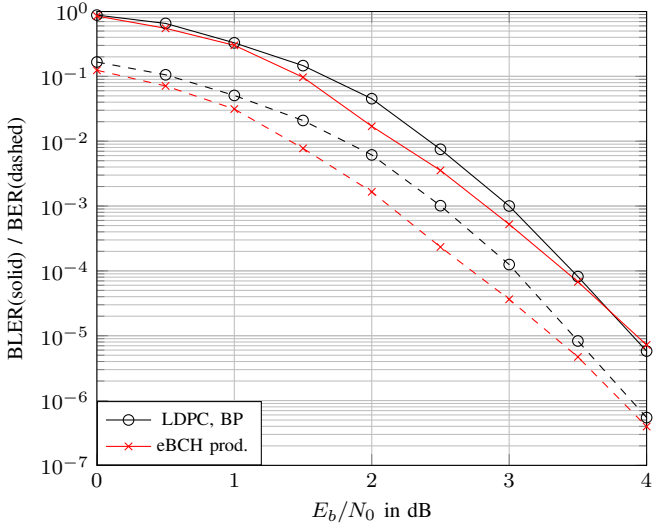

\begin{figure}
	\centering
	\footnotesize
	\begin{tikzpicture}[scale=1]
\footnotesize
\begin{semilogyaxis}[
legend style={at={(1,1)},anchor= north east, font=\scriptsize},
ymin=1e-7,
ymax=1,
width=3.5in,
height=3in,
grid=both,
xmin = 0,
xmax = 3,
xlabel = $E_b/N_0$ in dB,
ylabel = {BLER(solid) / BER(dashed)},
]

\addplot[black,mark=o]
table[]{x y
0 0.67568
0.5 0.35842
1 0.068166
1.5 0.011058
2 0.0012219
2.5 8.5235e-05
3 1.0756e-05
};\addlegendentry{LDPC, BP}

\addplot[red,mark=x]
table[]{x y
0 0.78556
0.5 0.40945
1 0.056128
1.5 0.0077405
2 0.00086483
2.5 0.000078374
3 7.1442e-06
};\addlegendentry{CRC prod.} 

\addplot[black,mark=o,dashed,mark options=solid]
table[]{x y
0 0.12072
0.5 0.056615
1 0.0090222
1.5 0.0013427
2 9.9438e-05
2.5 5.1444e-06
3 4.9088e-07
};

\addplot[red,mark=x,dashed,mark options=solid]
table[]{x y
0 0.13234
0.5 0.064752
1 0.0065473
1.5 0.00056472
2 4.9737e-05
2.5 2.4345e-06
3 2.8763e-07
};

\end{semilogyaxis}

\end{tikzpicture}

\begin{tikzpicture}[scale=1]
\begin{semilogyaxis}[
legend style={at={(1,1)},anchor= north east, font=\scriptsize},
ymin=1e1,
ymax=1e4,
width=3.5in,
height=1.5in,
grid=both,
xmin = 0,
xmax = 3,
ylabel = {avg. \# of Guesses per bit},
]

\addplot[red,mark=x]
table[]{x y
0 4335.9
0.5 2541.7
1 1148.8
1.5 758.33
2 579.5702
2.5 479.2741
3 400.126
};\addlegendentry{CRC prod.}

\addplot[red,mark=x,dashed,mark options=solid]
table[]{x y
0 352.87
0.5 204.52
1 93.396
1.5 61.019
2 46.6607
2.5 38.7546
3 32.9464
};\addlegendentry{CRC prod., parall.}

\end{semilogyaxis}

\end{tikzpicture}

\begin{tikzpicture}[scale=1]
\begin{semilogyaxis}[
legend style={at={(1,1)},anchor= north east, font=\scriptsize},
ymin=1e0,
ymax=1e2,
width=3.5in,
height=1.5in,
grid=both,
xmin = 0,
xmax = 3,
ylabel = {avg. \# of iterations},
]

\addplot[black,mark=o]
table[]{x y
0 38.125
0.5 29.688
1 20.243
1.5 12.95
2 8.3332
2.5 6.9155
3 5.9874
};\addlegendentry{LDPC, BP}

\addplot[red,mark=x]
table[]{x y
0 14.135
0.5 8.3235
1 3.6235
1.5 2.3328
2 1.7802
2.5 1.4766
3 1.2321
};\addlegendentry{CRC prod.}

\end{semilogyaxis}

\end{tikzpicture}
	\caption{\ac{AWGN} performance of the $(625, 225)$ 5G LDPC with max. number of iterations $I_\text{max}=50$ as compared to a $(625, 225)=(25,15)^2$ CRC \texttt{2b9} product code decoded using \ac{SOGRAND} with $\alpha=0.5$ and max. iteration $20$, where lists are added to until $L=4$ or the predicted list-BLER is below $10^{-5}$. Upper panel: \ac{BLER} and \ac{BER}. Middle panel: average number of queries per-bit until a decoding with \ac{SOGRAND}, where parallelized assumes all rows/columns are decoded in parallel. Lower panel: average number of iterations.}
	\label{fig:prod_grand_25_15}
\end{figure}

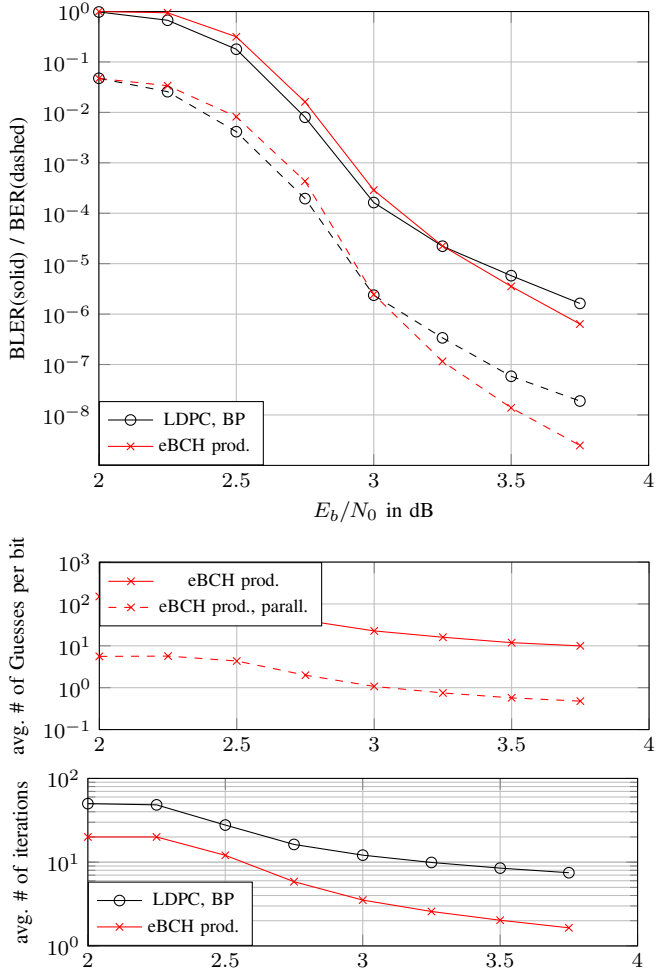
\begin{figure}
	\centering
	\footnotesize
	\begin{tikzpicture}[scale=1]
\footnotesize
\begin{semilogyaxis}[
legend style={at={(0,0)},anchor= south west, font=\scriptsize},
ymin=1e-9,
ymax=1,
width=3.5in,
height=3in,
ytick={1,0.1,0.01,0.001,0.0001,0.00001,0.000001,0.0000001,0.00000001},
grid=both,
xmin = 2,
xmax = 4,
xlabel = $E_b/N_0$ in dB,
ylabel = {BLER(solid) / BER(dashed)},
]


\addplot[black,mark=o]
table[]{x y
2 0.98039
2.25 0.67114
2.5 0.17889
2.75 0.0080173
3 0.00016304
3.25 2.219e-05
3.5 5.7831e-06
3.75 1.62929591361994e-06
};\addlegendentry{LDPC, BP}






\addplot[red,mark=x]
table[]{x y
2 1
2.25 0.94563
2.5 0.31531
2.75 0.016148
3 0.0002876
3.25 2.2434e-05
3.5 3.5442e-06
3.75 6.3561e-07
};\addlegendentry{eBCH prod.} 

\addplot[black,mark=o,dashed,mark options=solid]
table[]{x y
2 0.047405
2.25 0.025716
2.5 0.0041419
2.75 0.00019504
3 2.3697e-06
3.25 3.3772e-07
3.5 5.8494e-08
3.75 1.8893e-08
};


\addplot[red,mark=x,dashed,mark options=solid]
table[]{x y
2 0.047034
2.25 0.033847
2.5 0.0082076
2.75 0.000432
3 2.4417e-06
3.25 1.1546e-07
3.5 1.3845e-08
3.75 2.4829e-09
};

\end{semilogyaxis}

\end{tikzpicture}

\begin{tikzpicture}[scale=1]
\begin{semilogyaxis}[
legend style={at={(1,1)},anchor= north east, font=\scriptsize},
ymin=1e-1,
ymax=1e3,
width=3.5in,
height=1.5in,
ytick={1000,100,10,1,0.1,0.01,0.001,0.0001,0.00001,0.000001,0.0000001,0.00000001},
grid=both,
xmin = 2,
xmax = 4,
ylabel = {avg. \# of Guesses per bit},
]

\addplot[red,mark=x]
table[]{x y
2 146.96
2.25 145.22
2.5 104.9
2.75 49.782
3 29.69
3.25 22.3732
3.5 17.7476
3.75 14.798
};\addlegendentry{eBCH prod.}

\addplot[red,mark=x,dashed,mark options=solid]
table[]{x y
2 5.3045
2.25 5.6349
2.5 4.6403
2.75 2.2121
3 1.2879
3.25 0.96292
3.5 0.75039
3.75 0.6156
};\addlegendentry{eBCH prod., parall.} 

\end{semilogyaxis}

\end{tikzpicture}

\begin{tikzpicture}[scale=1]
\begin{semilogyaxis}[
legend style={at={(0,0)},anchor= south west, font=\scriptsize},
ymin=1e0,
ymax=1e2,
width=3.5in,
height=1.5in,
grid=both,
xmin = 2,
xmax = 4,
ylabel = {avg. \# of iterations},
]
\addplot[black,mark=o]
table[]{x y
2 50
2.25 48.333
2.5 27.75
2.75 16.28
3 12.128
3.25 9.9195
3.5 8.4932
3.75 7.4798
};\addlegendentry{LDPC, BP}

\addplot[red,mark=x]
table[]{x y
2 20
2.25 19.5
2.5 12.689
2.75 5.7449
3 3.348
3.25 2.498
3.5 1.9608
3.75 1.6269
};\addlegendentry{eBCH prod.}

\end{semilogyaxis}

\end{tikzpicture}
	\caption{AWGN performance of the $(4096, 3249)$ 5G LDPC with max. iterations $50$ compared to a $(4096, 3249)=(64,57)^2$ eBCH product code decoded using \ac{SOGRAND} with $\alpha=0.5$ and max. iteration $20$, where lists are added to until $L=4$ or the predicted list-BLER is below $10^{-6}$. Upper panel: \ac{BLER} and \ac{BER}. Middle panel: average number of queries per-bit until a decoding with \ac{SOGRAND}, where parallelized assumes all rows/columns are decoded in parallel. Lower panel: average number of iterations.}
	\label{fig:prod_grand_20_64}
\end{figure}

\begin{figure}
	\centering
	\footnotesize
	\begin{tikzpicture}[scale=1]
\footnotesize
\begin{semilogyaxis}[
legend style={at={(0,0)},anchor= south west, font=\scriptsize},
ymin=1e-7,
ymax=1,
width=3.5in,
height=3in,
grid=both,
xmin = 0,
xmax = 4,
xlabel = $E_b/N_0$ in dB,
ylabel = {BLER(solid) / BER(dashed)},
]

\addplot[black,mark=o]
table[]{x y
0 0.56604
0.5 0.28846
1 0.12821
1.5 0.054945
2 0.024135
2.5 0.0066167
3 0.00099295
3.5 0.00011574
4 1.6494e-05
};\addlegendentry{LDPC, BP}


\addplot[blue,mark=o]
table[]{x y
0 0.78947
0.5 0.625
1 0.375
1.5 0.19481
2 0.10345
2.5 0.027959
3 0.0057837
3.5 0.0008394
4 4.5241e-05
};\addlegendentry{polar, SCL} 


\addplot[red,mark=x]
table[]{x y
0 0.26549
0.5 0.17442
1 0.080429
1.5 0.025862
2 0.011966
2.5 0.0025875
3 0.00043492
3.5 7.5891e-05
4 8.1967e-06
};\addlegendentry{eBCH prod.}

\addplot[black,mark=o,dashed,mark options=solid]
table[]{x y
0 0.12707
0.5 0.067111
1 0.031136
1.5 0.012335
2 0.0056315
2.5 0.0015574
3 0.00021345
3.5 2.1888e-05
4 3.2461e-06
};


\addplot[blue,mark=o,dashed,mark options=solid]
table[]{x y
0 0.3898
0.5 0.3112
1 0.18652
1.5 0.097098
2 0.052425
2.5 0.013994
3 0.0028858
3.5 0.00042757
4 2.2691e-05
}; 


\addplot[red,mark=x,dashed,mark options=solid]
table[]{x y
0 0.063261
0.5 0.038336
1 0.014431
1.5 0.0047751
2 0.0019695
2.5 0.0004043
3 6.4557e-05
3.5 1.0329e-05
4 1.1211e-06
};

\end{semilogyaxis}

\end{tikzpicture}

\begin{tikzpicture}[scale=1]
\begin{semilogyaxis}[
legend style={at={(1,1)},anchor= north east, font=\scriptsize},
legend columns=2,
ymin=1e1,
ymax=1e3,
width=3.5in,
height=1.5in,
grid=both,
xmin = 0,
xmax = 4,
ylabel = {avg. \# of Guesses per bit},
]

\addplot[red,mark=x]
table[]{x y
0 454.96
0.5 381.35
1 308.62
1.5 258.74
2 224.15
2.5 197.92
3 174.8816
3.5 159.0196
4 143.2864
};\addlegendentry{eBCH prod.}

\addplot[red,mark=x,dashed,mark options=solid]
table[]{x y
0 48.79
0.5 40.456
1 32.451
1.5 27.098
2 23.407
2.5 20.646
3 18.2302
3.5 16.54
4 15.0449
};\addlegendentry{eBCH prod., parall.} 

\end{semilogyaxis}

\end{tikzpicture}

\begin{tikzpicture}[scale=1]
\begin{semilogyaxis}[
legend style={at={(0,0)},anchor= south west, font=\scriptsize},
legend columns=2,
ymin=1e-1,
ymax=1e2,
width=3.5in,
height=1.5in,
grid=both,
xmin = 0,
xmax = 4,
ylabel = {avg. \# of iterations},
]

\addplot[black,mark=o]
table[]{x y
0 33.679
0.5 22.317
1 14.684
1.5 10.701
2 8.2647
2.5 6.4618
3 5.5573
3.5 4.9888
4 4.5823
};\addlegendentry{LDPC, BP}


\addplot[red,mark=x]
table[]{x y
0 3.5362
0.5 2.9099
1 2.3045
1.5 1.9036
2 1.632
2.5 1.4355
3 1.2563
3.5 1.1442
4 1.0438
};\addlegendentry{eBCH prod.} 

\end{semilogyaxis}

\end{tikzpicture}
	\caption{AWGN performance of the $(256, 49)$ 5G LDPC with max. iterations $50$ and the $(256, 49)$ 5G \ac{CA-Polar} ($24$-bits CRC) decoded with \ac{CA-SCL} ($L=16$) compared to a $(256, 49)=(16,7)^2$ eBCH product code decoded using \ac{SOGRAND} with $\alpha=0.5$ and max. iteration $20$, where lists are added to until $L=4$ or the predicted list-BLER is below $10^{-5}$. Upper panel: \ac{BLER} and \ac{BER}. Middle panel: average number of queries per-bit until a decoding with \ac{SOGRAND}, where parallelized assumes all rows/columns are decoded in parallel. Lower panel: average number of iterations.}
	\label{fig:prod_grand_256_49}
\end{figure}

For iterative decoding, we construct two types of long codes based on short $(n,k)$ component codes: $(n^2,k^2)$ product codes~\cite{elias_error-free_1954} and $(n^2,n^2-2n(n-k))$ \ac{QC}-GLDPCs~\cite{Lentmaier10} with adjacency matrix
\begin{align*}
\begin{bmatrix}
I_{n}^{(0)} & I_{n}^{(0)} & \cdots &I_{n}^{(0)}\\
I_{n}^{(0)} & I_{n}^{(1)} & \cdots &I_{n}^{(n-1)}
\end{bmatrix},
\end{align*}
where $I_{n}^{(i)}$ a $n\times n$ circulant permutation matrix obtained by the right rotation (by 1 position) of the identity matrix. Product codes can be thought of as a special case of the QC-GLDPC ensemble with variable node degree-$2$.

\begin{figure*}[ht]
	\centering
	\begin{tikzpicture}[scale=0.5]
    \footnotesize
\draw[dashed,rounded corners=2pt] (-0.6,-0.6) rectangle (8.1,7.6);
\draw[gray] (-0.1,-0.1) rectangle (3.1,3.1);
\draw[] (4.5,0) rectangle (7.5,3);
\draw[] (0,4) rectangle (3,7);
\draw[] (4.5,4) rectangle (7.5,7);
\node at (8.25,7.9) {row update};
\node at (3.75,-1) {$(1)$};

\draw[fill=red!20] (0,2.5) rectangle (3,3);
\node at (1.5,2.75) {\scalebox{0.5}{SOGRAND}};
\draw[fill=blue!20] (0,1.9) rectangle (3,2.4);
\node at (1.5,2.15) {\scalebox{0.5}{SOGRAND}};
\draw[fill=brown!20] (0,1.3) rectangle (3,1.8);
\node at (1.5,1.55) {\scalebox{0.5}{SOGRAND}};
\node at (1.5,1) {\scalebox{0.5}{$\vdots$}};

\draw[fill=red!50] (0,6.6) rectangle (3,7);
\draw[fill=blue!50] (0,6.1) rectangle (3,6.5);
\draw[fill=brown!50] (0,5.6) rectangle (3,6);
\node at (1.5,5) {\scalebox{0.5}{$\mathbf{L}_\text{Ch}$}};

\node at (6,5.5) {\scalebox{0.5}{Hard Decision}};

\draw[fill=red!50] (4.5,2.6) rectangle (7.5,3);
\draw[fill=blue!50] (4.5,2.1) rectangle (7.5,2.5);
\draw[fill=brown!50] (4.5,1.6) rectangle (7.5,2);
\node at (6,1) {\scalebox{0.5}{$\mathbf{L}_\text{A}/\mathbf{L}_\text{E}$}};

\draw[rounded corners=2pt=2pt,->,red] (0,6.8) -- (-0.2,6.8) -- (-0.2,2.9) -- (0,2.9);
\draw[rounded corners=2pt,->,blue] (0,6.3) -- (-0.3,6.3) -- (-0.3,2.3) -- (0,2.3);
\draw[rounded corners=2pt,->,brown] (0,5.8) -- (-0.4,5.8) -- (-0.4,1.7) -- (0,1.7);

\draw[dotted] (-0.3,3.5) ellipse (0.25cm and 0.4cm);
\node at (0.3,3.5)[align=center] {\scalebox{0.5}{$L_\text{Ch}$}};

\draw[rounded corners=2pt,->,red] (7.5,2.8) -- (7.7,2.8) -- (7.7,-0.2) -- (-0.2,-0.2) -- (-0.2,2.6) -- (0,2.6);
\draw[rounded corners=2pt,->,blue] (7.5,2.3) -- (7.8,2.3) -- (7.8,-0.3) -- (-0.3,-0.3) -- (-0.3,2) -- (0,2);
\draw[rounded corners=2pt,->,brown] (7.5,1.8) -- (7.9,1.8) -- (7.9,-0.4) -- (-0.4,-0.4) -- (-0.4,1.4) -- (0,1.4);

\draw[dotted] (3.75,-0.3) ellipse (0.4cm and 0.25cm);
\node at (3.75,0.4)[align=center] {\scalebox{0.5}{$L_\text{A}$}};

\draw[dashed,rounded corners=2pt] (8.4,-0.6) rectangle (17.1,7.6);
\draw[gray] (8.9,-0.1) rectangle (12.1,3.1);
\draw[] (13.5,0) rectangle (16.5,3);
\draw[] (9,4) rectangle (12,7);
\draw[] (13.5,4) rectangle (16.5,7);
\node at (12.75,-1) {$(2)$};

\draw[fill=red!20] (9,2.5) rectangle (12,3);
\node at (10.5,2.75) {\scalebox{0.5}{SOGRAND}};
\draw[fill=blue!20] (9,1.9) rectangle (12,2.4);
\node at (10.5,2.15) {\scalebox{0.5}{SOGRAND}};
\draw[fill=brown!20] (9,1.3) rectangle (12,1.8);
\node at (10.5,1.55) {\scalebox{0.5}{SOGRAND}};
\node at (10.5,1) {\scalebox{0.5}{$\vdots$}};

\node at (10.5,5.5) {\scalebox{0.5}{$\mathbf{L}_\text{Ch}$}};
\draw[fill=red!50] (13.5,6.6) rectangle (16.5,7);
\draw[fill=blue!50] (13.5,6.1) rectangle (16.5,6.5);
\draw[fill=brown!50] (13.5,5.6) rectangle (16.5,6);
\node at (15,5) {\scalebox{0.5}{Hard Decision}};

\draw[fill=red!50] (13.5,2.6) rectangle (16.5,3);
\draw[fill=blue!50] (13.5,2.1) rectangle (16.5,2.5);
\draw[fill=brown!50] (13.5,1.6) rectangle (16.5,2);
\node at (15,1) {\scalebox{0.5}{$\mathbf{L}_\text{A}/\mathbf{L}_\text{E}$}};

\draw[rounded corners=2pt,->,red] (12,2.9) -- (12.2,2.9) -- (12.2,6.8) -- (13.5,6.8);
\draw[rounded corners=2pt,->,blue] (12,2.3) -- (12.3,2.3) -- (12.3,6.3) -- (13.5,6.3);
\draw[rounded corners=2pt,->,brown] (12,1.7) -- (12.4,1.7) -- (12.4,5.8) -- (13.5,5.8);
\draw[dotted] (13,6.3) ellipse (0.4cm and 0.8cm);
\node at (13,5)[align=center,scale=0.5] {hard\\decision\\of $L_\text{APP}$};

\draw[rounded corners=2pt,->,red] (12,2.6) -- (12.9,2.6) -- (12.9,2.8) -- (13.5,2.8);
\draw[rounded corners=2pt,->,blue] (12,2) -- (13,2) -- (13,2.3) -- (13.5,2.3);
\draw[rounded corners=2pt,->,brown] (12,1.4) -- (13.1,1.4) -- (13.1,1.8) -- (13.5,1.8);
\draw[dotted] (13,2.15) ellipse (0.4cm and 1cm);
\node at (12.8,0.7)[align=center,scale=0.5] {overwrite\\$L_\text{E}$};

\end{tikzpicture}
    \begin{tikzpicture}[scale=0.5]
    \footnotesize
\draw[dashed,rounded corners=2pt] (-0.6,-0.6) rectangle (8.1,7.6);
\draw[gray] (-0.1,-0.1) rectangle (3.1,3.1);
\draw[] (4.5,0) rectangle (7.5,3);
\draw[] (0,4) rectangle (3,7);
\draw[] (4.5,4) rectangle (7.5,7);
\node at (8.25,7.9) {column update};
\node at (3.75,-1) {$(1)$};

\draw[fill=red!20] (0,2.5) rectangle (3,3);
\node at (1.5,2.75) {\scalebox{0.5}{SOGRAND}};
\draw[fill=blue!20] (0,1.9) rectangle (3,2.4);
\node at (1.5,2.15) {\scalebox{0.5}{SOGRAND}};
\draw[fill=brown!20] (0,1.3) rectangle (3,1.8);
\node at (1.5,1.55) {\scalebox{0.5}{SOGRAND}};
\node at (1.5,1) {\scalebox{0.5}{$\vdots$}};

\draw[fill=red!50] (0,4) rectangle (0.4,7);
\draw[fill=blue!50] (0.5,4) rectangle (0.9,7);
\draw[fill=brown!50] (1,4) rectangle (1.4,7);
\node at (2,5.5) {\scalebox{0.5}{$\mathbf{L}_\text{Ch}$}};

\node at (6,5.5) {\scalebox{0.5}{Hard Decision}};

\draw[fill=red!50] (4.5,0) rectangle (4.9,3);
\draw[fill=blue!50] (5,0) rectangle (5.4,3);
\draw[fill=brown!50] (5.5,0) rectangle (5.9,3);
\node at (6.5,1.5) {\scalebox{0.5}{$\mathbf{L}_\text{A}/\mathbf{L}_\text{E}$}};

\draw[rounded corners=2pt,->,red] (0.2,7) -- (0.2,7.2) -- (-0.2,7.2) -- (-0.2,2.9) -- (0,2.9);
\draw[rounded corners=2pt,->,blue] (0.7,7) -- (0.7,7.3) -- (-0.3,7.3) -- (-0.3,2.3) -- (0,2.3);
\draw[rounded corners=2pt,->,brown] (1.2,7) -- (1.2,7.4) -- (-0.4,7.4) -- (-0.4,1.7) -- (0,1.7);

\draw[dotted] (-0.3,3.5) ellipse (0.25cm and 0.4cm);
\node at (0.3,3.5)[align=center] {\scalebox{0.5}{$L_\text{Ch}$}};

\draw[rounded corners=2pt,->,red] (4.7,0) -- (4.7,-0.2) -- (-0.2,-0.2) -- (-0.2,2.6) -- (0,2.6);
\draw[rounded corners=2pt,->,blue] (5.2,0) -- (5.2,-0.3) -- (-0.3,-0.3) -- (-0.3,2) -- (0,2);
\draw[rounded corners=2pt,->,brown] (5.7,0) -- (5.7,-0.4) -- (-0.4,-0.4) -- (-0.4,1.4) -- (0,1.4);

\draw[dotted] (3.75,-0.3) ellipse (0.4cm and 0.25cm);
\node at (3.75,0.4)[align=center] {\scalebox{0.5}{$L_\text{A}$}};

\draw[dashed,rounded corners=2pt] (8.4,-0.6) rectangle (17.1,7.6);
\draw[gray] (8.9,-0.1) rectangle (12.1,3.1);
\draw[] (13.5,0) rectangle (16.5,3);
\draw[] (9,4) rectangle (12,7);
\draw[] (13.5,4) rectangle (16.5,7);
\node at (12.75,-1) {$(2)$};

\draw[fill=red!20] (9,2.5) rectangle (12,3);
\node at (10.5,2.75) {\scalebox{0.5}{SOGRAND}};
\draw[fill=blue!20] (9,1.9) rectangle (12,2.4);
\node at (10.5,2.15) {\scalebox{0.5}{SOGRAND}};
\draw[fill=brown!20] (9,1.3) rectangle (12,1.8);
\node at (10.5,1.55) {\scalebox{0.5}{SOGRAND}};
\node at (10.5,1) {\scalebox{0.5}{$\vdots$}};

\node at (10.5,5.5) {\scalebox{0.5}{$\mathbf{L}_\text{Ch}$}};

\draw[fill=red!50] (13.5,4) rectangle (13.9,7);
\draw[fill=blue!50] (14,4) rectangle (14.4,7);
\draw[fill=brown!50] (14.5,4) rectangle (14.9,7);
\node at (15.5,5.5)[align=center,scale=0.5] {Hard\\Decision};

\draw[fill=red!50] (13.5,0) rectangle (13.9,3);
\draw[fill=blue!50] (14,0) rectangle (14.4,3);
\draw[fill=brown!50] (14.5,0) rectangle (14.9,3);
\node at (15.5,1.5) {\scalebox{0.5}{$\mathbf{L}_\text{A}/\mathbf{L}_\text{E}$}};

\draw[rounded corners=2pt,->,red] (12,2.9) -- (12.2,2.9) -- (12.2,7.4) -- (13.7,7.4) --(13.7,7);
\draw[rounded corners=2pt,->,blue] (12,2.3) -- (12.3,2.3) -- (12.3,7.3) -- (14.2,7.3) -- (14.2,7);
\draw[rounded corners=2pt,->,brown] (12,1.7) -- (12.4,1.7) -- (12.4,7.2) -- (14.7,7.2) -- (14.7,7);

\draw[dotted] (13,7.3) ellipse (0.4cm and 0.25cm);
\node at (13,6.5)[align=center,scale=0.5] {hard\\decision\\of $L_\text{APP}$};

\draw[rounded corners=2pt,->,red] (12,2.6) -- (12.8,2.6) -- (12.8,3.4) -- (13.7,3.4) -- (13.7,3);
\draw[rounded corners=2pt,->,blue] (12,2) -- (12.9,2) -- (12.9,3.3) -- (14.2,3.3) -- (14.2,3);
\draw[rounded corners=2pt,->,brown] (12,1.4) -- (13.1,1.4) -- (13,3.2) -- (14.7,3.2) -- (14.7,3);

\draw[dotted] (13,2.15) ellipse (0.4cm and 1cm);
\node at (12.8,0.7)[align=center, scale=0.5] {overwrite\\$L_\text{E}$};

\end{tikzpicture}
	\caption{Block turbo decoding of product codes with SOGRAND.}
	\label{fig:turboDec}
\end{figure*}
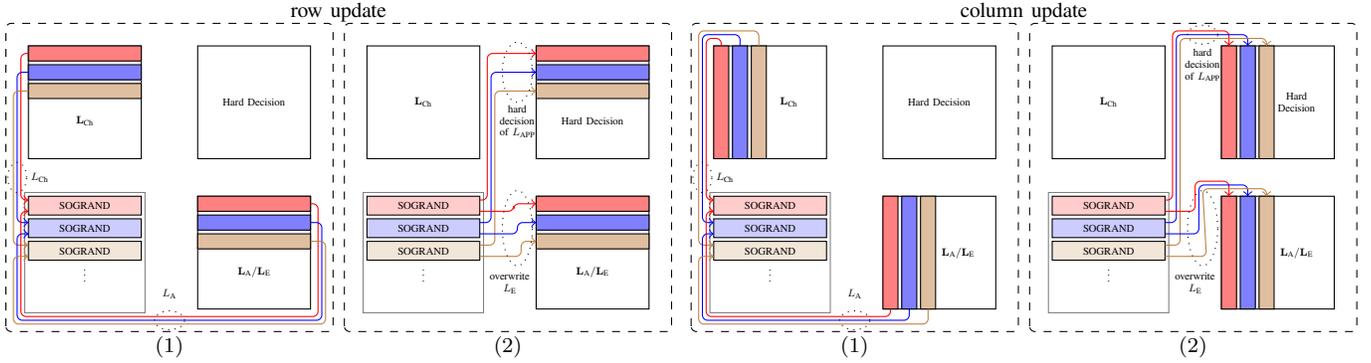

Block turbo decoding of product codes works as follows (see Fig.~\ref{fig:turboDec}):
\begin{itemize}
    \item[0] The channel LLRs are stored in an $n\times n$ matrix $\mathbf{L}_\text{Ch}$. The a priori LLRs of the coded bits are initialized to zero, i.e., $\mathbf{L}_\text{A}=\mathbf{0}$.
    \item[1] Each row of $\mathbf{L}_\text{Ch}+\mathbf{L}_\text{A}$ is processed by an \ac{SISO} decoder, and the resulting APP LLRs and extrinsic LLRs are stored in the corresponding rows of $\mathbf{L}_\text{APP}$ and $\mathbf{L}_\text{E}$, respectively. A hard decision is made based on $\mathbf{L}_\text{APP}$. If all rows and columns of the decision correspond to valid codewords, the block turbo decoder returns the hard output, indicating successful decoding. Otherwise, $\mathbf{L}_\text{A}$ is set to $\alpha\mathbf{L}_\text{E}$, for some $\alpha>0$, and the decoder proceeds to the column update. 
    \item[2] Each column of $\mathbf{L}_\text{Ch}+\mathbf{L}_\text{A}$ is decoded and the columns of $\mathbf{L}_\text{APP}$ and $\mathbf{L}_\text{E}$ are updated as step $1$. A hard decision is performed using $\mathbf{L}_\text{APP}$. If the obtained binary matrix is valid, decoding success is declared. If the maximum iteration count is reached, a decoding failure is returned. Otherwise, we set $\mathbf{L}_\text{A} = \alpha\mathbf{L}_\text{E}$ and proceed to the next iteration (i.e., return to step 1).
\end{itemize}
In simulation, we set $\alpha=0.5$ and use 1-line ORBGRAND in conjunction with eq. \eqref{eq:LLRi} as the SISO \ac{SOGRAND} decoder. With even codes the demodulated signal is used to identify the parity of the noise effect that is being sought.

Decoding of GLDPC codes is based on the \ac{BP} principle over the Tanner graph which is a generalization of the block turbo decoding of product codes. Each decoding iteration consists of an exchange of soft messages between \acp{VN} and \acp{CN}. Each \ac{CN} works as an SISO decoder as described in Sec.~\ref{sec:siso}, i.e., each CN treats the extrinsic LLRs from last iteration as a-priori information and generates the corresponding extrinsic LLRs. Again, for \ac{SOGRAND} 1-line \ac{ORBGRAND} with eq. \eqref{eq:LLRi} is used. In simulation, the outgoing extrinsic LLRs are always weighted with the scaling factor $\alpha=0.5$. At the end of each iteration, an \ac{APP} LLR is calculated for each coded bit, and a hard decision is made. If the resulting binary sequence is a valid codeword, then decoding is declared complete. 

To evaluate performance, we consider both \ac{AWGN} and Rician fading channels. In upper panels, we report \ac{BLER} and \ac{BER} metrics for decoding accuracy. In middle panels, for \ac{GRAND}-based algorithms we report the average number of queries per coded bit until a decoding has been found and, in the presence of parallelized decoding of rows and columns, the the average sum of the maximum number of queries per-bit, which serve as a proxy for decoding complexity and energy \cite{duffy2022_ordered,an2023soft,Riaz23}, where multiple queries can be made per clock-cycle in hardware. In lower panels, we show the average number of iterations until a decoding is found, which serves as a proxy for latency. In keeping with the existing use for \ac{LDPC} code decoding, each half-iteration is the largest entirely parallelized unit. For product codes, all rows can be decoded in parallel and all columns can be decoded in parallel. 

\subsection{AWGN Channels}
The results in Fig. \ref{fig:prod_grand_32_21} show that Elias's original product code can give better \ac{BLER} performance than a 5G \ac{LDPC} code when decoded with the new \ac{SOGRAND}. We first demonstrate that this holds consistently for other code dimensions, where all \ac{LDPC} codes are from the 5G  standard. Fig. \ref{fig:prod_grand_16_11} shows a comparison for $(256, 121)$ \ac{LDPC} and a $(16,11)^2$ \ac{eBCH} product code, where it can be seen that the product code outperforms the \ac{LDPC} in terms of both \ac{BLER} and \ac{BER}. The middle plot shows the average number of \ac{SOGRAND} codebook queries per-bit per product code decoding,  showing reduced effort for improved decoding. Note that each iteration of the 1-line \ac{ORBGRAND} decoding of a product code all rows or column component codes can be decoded entirely in parallel. In that setting, the \enquote{parallelized} results in Fig.~\ref{fig:prod_grand_16_11} show the average sum of the maximum number of queries per-bit required for each parallel decoding of all rows or columns, which results in significantly reduced decoding latency. Furthermore, that \ac{GRAND} codebook queries are themselves parallelizable, which has been exploited in in-silicon realizations where multiple queries are made per clock-cycle \cite{Riaz21,Riaz23}. The lower plot demonstrates that the average number of iterations required to identify a decoding is dramactically smaller for the product code than the \ac{LDPC} code, strongly indicating lower latency decoding.

As \ac{GRAND} algorithms can decode any moderate redundancy code, the product code construction is not confined to the \ac{dRM} and \ac{eBCH} codes reported in Figs \ref{fig:prod_grand_32_21} and \ref{fig:prod_grand_16_11}. Fig. \ref{fig:prod_grand_25_15} provides a further example where a $(625, 225)$ \ac{LDPC} code is compared with a $(25,15)^2$ product code that uses a simple \ac{CRC} code as its component code. Again, the product code outperforms the \ac{LDPC} code with fewer iterations, minimal complexity, and extremely low latency when parallelized. Fig. \ref{fig:prod_grand_20_64} provides a further comparison for the $(4096,3249)$ \ac{LDPC} code and a $(64,57)^2$ \ac{eBCH} product code, resulting in the same conclusion. Fig.~\ref{fig:prod_grand_256_49} provides a further example where a $(256, 49)$ \ac{LDPC} code is compared with a $(16,7)^2$ eBCH product code. Furthermore, the performance of a $(256,49)$ 5G \ac{CA-Polar} code with $24$-bit CRC decoded with \ac{CA-SCL} is also provided. As CA-SCL's complexity increases exponentially with list size, when implemented in hardware, typical list sizes are 2, 4 or 8 \cite{liang2016hardware,Tao_CASCL_21,Kam_CASCL_22}. Here we allow CA-SCL a generous $L=16$ to extract improved performance. Thus, when decoded in an iterative fashion with 1-line \ac{ORBGRAND} and the new \ac{SO} as \ac{SOGRAND}, Elias's product codes offer comparable or better performance than the \ac{LDPC} codes selected for 5G \ac{NR}. For relatively short and low-rate codes, such as $(256,49)$, the proposed approach performs slightly better than the CA-Polar code decoded with CA-SCL. 

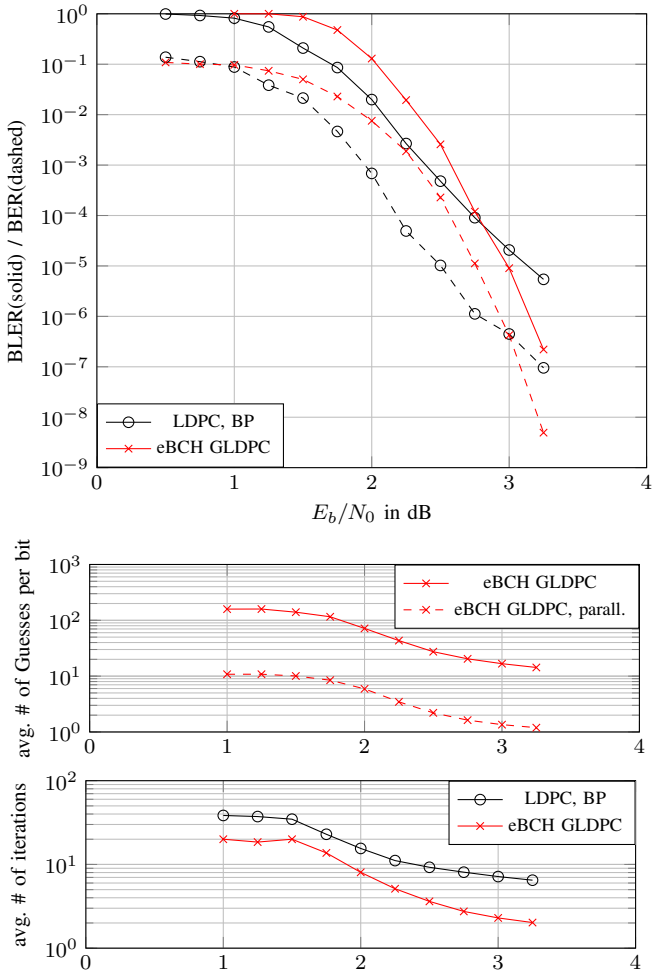
\begin{figure}
	\centering
	\footnotesize
	\begin{tikzpicture}[scale=1]
\footnotesize
\begin{semilogyaxis}[
legend style={at={(0,0)},anchor= south west, font=\scriptsize},
ymin=1e-9,
ymax=1,
width=3.5in,
height=3in,
ytick={1,0.1,0.01,0.001,0.0001,0.00001,0.000001,0.0000001,0.00000001,0.000000001},
grid=both,
xmin = 0,
xmax = 4,
xlabel = $E_b/N_0$ in dB,
ylabel = {BLER(solid) / BER(dashed)},
]

\addplot[black,mark=o]
table[]{x y
0.5 0.9901
0.75 0.92593
1 0.81967
1.25 0.54945
1.5 0.20921
1.75 0.085911
2 0.019916
2.25 0.002673
2.5 0.00047986
2.75 9.0555e-05
3 2.0705e-05
3.25 5.4220e-06
};\addlegendentry{LDPC, BP}

\addplot[red,mark=x]
table[]{x y
1 1
1.25 1
1.5 0.87432
1.75 0.47684
2 0.12958
2.25 0.019432
2.5 0.0025784
2.75 0.00011942
3 9.0231e-06
3.25 2.2023e-07
};\addlegendentry{eBCH GLDPC}

\addplot[black,mark=o,dashed,mark options=solid]
table[]{x y
0.5 0.13772
0.75 0.11241
1 0.087996
1.25 0.038418
1.5 0.021307
1.75 0.0046404
2 0.00068321
2.25 0.000049349
2.5 1.0236e-05
2.75 1.1211e-06
3 4.4719e-07
3.25 9.528e-08
};

\addplot[red,mark=x,dashed,mark options=solid]
table[]{x y
0.5 0.10869
0.75 0.10015
1 0.097021
1.25 0.074072
1.5 0.050078
1.75 0.022896
2 0.0075559
2.25 0.0018799
2.5 0.00022785
2.75 1.1217e-05
3 4.2255e-07
3.25 4.9465e-09
};

\end{semilogyaxis}

\end{tikzpicture}

\begin{tikzpicture}[scale=1]
\begin{semilogyaxis}[
legend style={at={(1,1)},anchor= north east, font=\scriptsize},
ymin=1e0,
ymax=1e3,
width=3.5in,
height=1.5in,
grid=both,
xmin = 0,
xmax = 4,
ylabel = {avg. \# of Guesses per bit},
]

\addplot[red,mark=x]
table[]{x y
1 158.71
1.25 156.04
1.5 119.34
1.75 102.92
2 58.669
2.25 36.488
2.5 27.086
2.75 22.393
3 19.636
3.25 17.3544
};\addlegendentry{eBCH GLDPC}

\addplot[red,mark=x,dashed,mark options=solid]
table[]{x y
1 10.811
1.25 10.828
1.5 8.8984
1.75 7.6959
2 4.3769
2.25 2.6372
2.5 1.9232
2.75 1.5715
3 1.375
3.25 1.2334
};\addlegendentry{eBCH GLDPC, parall.}

\end{semilogyaxis}

\end{tikzpicture}

\begin{tikzpicture}[scale=1]
\begin{semilogyaxis}[
legend style={at={(1,1)},anchor= north east, font=\scriptsize},
ymin=1e0,
ymax=1e2,
width=3.5in,
height=1.5in,
grid=both,
xmin = 0,
xmax = 4,
ylabel = {avg. \# of iterations},
]

\addplot[black,mark=o]
table[]{x y
1 38.444
1.25 37.308
1.5 34.7
1.75 22.92
2 15.527
2.25 11.116
2.5 9.2476
2.75 8.0831
3 7.1669
3.25 6.4853
};\addlegendentry{LDPC, BP}

\addplot[red,mark=x]
table[]{x y
1 19.81
1.25 19.476
1.5 13.846
1.75 11.855
2 6.5644
2.25 4.0362
2.5 2.9682
2.75 2.4396
3 2.1321
3.25 1.9229
};\addlegendentry{eBCH GLDPC}

\end{semilogyaxis}

\end{tikzpicture}
	\caption{\ac{AWGN} performance of the $(1024, 640)$ 5G LDPC with max. iteration $50$ as compared to a $(1024,640)$ \ac{GLDPC} code \cite{Lentmaier10} with eBCH nodes using \ac{SOGRAND} with $\alpha=0.6$ and max. iteration $20$, where lists are added to until $L=4$ or the predicted list-BLER  is below $10^{-5}$. Upper panel: \ac{BLER} and \ac{BER}. Middle panel: average number of queries per-bit until a decoding with \ac{SOGRAND}, where parallelized assumes all rows/columns are decoded in parallel. Lower panel: average number of iterations.}
	\label{fig:gldpc_GRAND}
\end{figure}

While the \ac{LDPC} codes used in 5G trade waterfall for an error floor, variants of product codes called \ac{GLDPC} codes have been developed that have better minimum distance and much lower error floors. As with product codes, \ac{GLDPC} code components can also be decoded in parallel, which would result in low-latency decoding. Fig.~\ref{fig:gldpc_GRAND} provides results for one such example when a \ac{GLDPC} developed in \cite{Lentmaier10} is decoded with \ac{SOGRAND}. The \ac{GLDPC} results in a significantly steeper waterfall \ac{BLER} curve with a significantly lower error floor, at the expense of slightly degraded performance at lower SNR, offering more design possibilities for future systems. Again, much fewer iterations are needed for decoding of the \ac{GLDPC} than the \ac{LDPC}.

\subsection{Block Fading Channels}
While the results so far have been for \ac{AWGN} channels, the same conclusions are found for fading channels, as we demonstrate here. Consider a complex-alphabet block fading \ac{AWGN} channel of coherence block length $N_c$. In each coherence block, the transmitted symbols will be affected by the same fading coefficient $H$,
$Y_i = H X_i + Z_i,~i=1,\dots,N_c$, where ${Z}_i\sim\mathcal{CN}\left(0,{2\sigma^2}\right)$. The noise power spectral density is defined to be $N_0= 2\sigma^2$. If the receiver knows $H=h$ and, with $h^*$ being the complex conjugate of $h$, computes
\begin{align*}
\frac{1}{h}Y_i=\frac{h^*}{\left|h\right|^2}Y_i=X_i+\frac{h^*}{\left|h\right|^2}Z_i,
\end{align*}
then they can form the equivalent channel
$\tilde{Y}_i=X_i+\tilde{Z}_i$,
where
$\tilde{Y}_i={{h^*}Y_i}/{\left|h\right|^2}$ and 
$\tilde{Z}_i\sim\mathcal{CN}\left(0,{2\sigma^2}/{|h|^2}\right)$.
Suppose that the channel coefficient is estimated as $\hat{h}$ and Gray labeled \ac{QPSK} $x=f_\text{QPSK}\left(c_1c_2\right)$ is used, i.e., $\mathcal{X}=\left\{\pm\Delta\pm\Delta j\right\}$. The \acp{LLR} of $c_1,c_2$ are given by
\begin{align*}
\log\frac{ p_{Y|C_1,H}(y|0,\hat{h}) }{p_{Y|C_1,H}(y|1,\hat{h})} & = \frac{2\Delta\Re(y\,\hat{h}^*)}{\sigma^2}\\
\log\frac{ p_{Y|C_2,H}(y|0,\hat{h}) }{p_{Y|C_2,H}(y|1,\hat{h})} & = \frac{2\Delta\Im(y\,\hat{h}^*)}{\sigma^2}.
\end{align*}
The SNR is then defined to be
\begin{align*}
\frac{E_s}{N_0} = \frac{\text{E}\left[\|X\|^2\right]\,\text{E}\left[\|H\|^2\right]}{2\sigma^2}= \frac{\Delta^2\,\text{E}\left[\|H\|^2\right]}{\sigma^2}.
\end{align*}
We consider standard block Rician fading with factor-$K$~\cite{5271272,stuber2001principles}, $Y_i = H_\text{Ri} X_i + Z_i$. The fading coefficient of Rician fading is defined by $H_\text{Ri} = \sqrt{K/(K+1)} + \sqrt{1/(K+1)}H_\text{Ra}$,
where $H_\text{Ra}$ is the fading coefficient of Rayleigh fading, i.e., the real part and the imaginary part of $H_\text{Ra}$ are independently Gaussian with zero mean and variance $\sigma_\text{Ra}^2$, i.e., $H_\text{Ra}\sim\mathcal{CN}\left(0,2\sigma_\text{Ra}^2\right)$. A high $K$ factor suggests a more deterministic component in the channel, often due to the predominance of line-of-sight transmission, and as $K$ becomes very large, the channel approximates an \ac{AWGN} scenario. Conversely, a low $K$ factor indicates small-scale fading, usually resulting from a heavily scattered environment, and the channel transitions to Rayleigh fading when $K$ is zero. Without loss of generality, we set $\sigma_\text{Ra}^2=0.5$ such that
$\text{E}\left[\|H_\text{Ra}\|^2\right]=\text{E}\left[\|H_\text{Ri}\|^2\right]=1.$
The codewords are randomly interleaved prior to mapping to \ac{QPSK} symbols. In each codeword frame of length $N$ bits, there are $t$ distinct coherence blocks, i.e., $N = 2tN_c$.

\begin{figure}
	\centering
	\footnotesize
	\begin{tikzpicture}[scale=1]
\footnotesize
\begin{semilogyaxis}[
legend style={at={(0,0)},anchor= south west, font=\scriptsize},
ymin=1e-5,
ymax=1,
width=3.5in,
height=3in,
grid=both,
xmin = 0,
xmax = 8,
xlabel = $E_b/N_0$ in dB,
ylabel = {BLER(solid) / BER(dashed)},
]

\addplot[black,mark=o]
table[]{x y
0 0.78125
1 0.41667
2 0.20778
3 0.081169
4 0.02867
5 0.0081296
6 0.0023928
7 0.00054727
8 0.00013016
};\addlegendentry{LDPC, BP}

\addplot[blue,mark=o]
table[]{x y
0 0.8
1 0.57143
2 0.30037
3 0.13638
4 0.057307
5 0.02401
6 0.0075019
7 0.0028245
8 0.0012542
};\addlegendentry{dRM prod., ORBGRAND Pyndiah}

\addplot[red,mark=x]
table[]{x y
0 0.83333
1 0.47619
2 0.25126
3 0.10482
4 0.025641
5 0.007367
6 0.0015798
7 0.00033922
8 0.00007129
};\addlegendentry{dRM prod., SOGRAND} 

\addplot[black,mark=o,dashed,mark options=solid]
table[]{x y
0 0.1717
1 0.077343
2 0.037202
3 0.014437
4 0.0050318
5 0.0013398
6 0.00041996
7 9.6939e-05
8 2.5283e-05
};

\addplot[blue,mark=o,dashed,mark options=solid]
table[]{x y
0 0.21445
1 0.16278
2 0.08456
3 0.041084
4 0.017086
5 0.0074045
6 0.0023956
7 0.00090747
8 0.00040336
};

\addplot[red,mark=x,dashed,mark options=solid]
table[]{x y
0 0.17184
1 0.089332
2 0.043847
3 0.017938
4 0.0045117
5 0.0012458
6 0.00026952
7 5.9608e-05
8 1.0647e-05
};

\end{semilogyaxis}

\end{tikzpicture}

\begin{tikzpicture}[scale=1]
\begin{semilogyaxis}[
legend style={at={(1,1)},anchor= south east, font=\tiny},
legend columns=2,
ymin=1e0,
ymax=1e4,
width=3.5in,
height=1.5in,
ytick={10000,1000,100,10,1,0.1,0.01,0.001,0.0001,0.00001,0.000001,0.0000001,0.00000001},
grid=both,
xmin = 0,
xmax = 8,
ylabel = {avg. \# of Guesses per bit},
]

\addplot[red,mark=x]
table[]{x y
0 4139.5
1 2306.9
2 1496.6
3 745.35
4 391.54
5 269.76
6 168.14
7 116.37
8 73.5282
};\addlegendentry{dRM prod., SOGRAND} 

\addplot[red,mark=x,dashed,mark options=solid]
table[]{x y
0 277.84
1 146.67
2 98.707
3 51.411
4 28.028
5 20.247
6 13.61
7 10.103
8 7.4731
};\addlegendentry{dRM prod., SOGRAND, parall.}

\addplot[blue,mark=x]
table[]{x y
0 4229.5
1 3391.2
2 1833.1
3 1198.6
4 673.42
5 469.04
6 357.2
7 299.8
8 286.8771
};\addlegendentry{dRM prod., ORBGRAND Pyndiah} 

\addplot[blue,mark=x,dashed,mark options=solid]
table[]{x y
0 266.15
1 206.25
2 116.37
3 76.379
4 44.754
5 31.46
6 24.45
7 20.944
8 20.0704
};\addlegendentry{dRM prod., ORBGRAND Pyndiah, parall.} 

\end{semilogyaxis}

\end{tikzpicture}

\begin{tikzpicture}[scale=1]
\begin{semilogyaxis}[
legend style={at={(0,0)},anchor= south west, font=\tiny},
ymin=1e-1,
ymax=1e2,
width=3.5in,
height=1.5in,
grid=both,
xmin = 0,
xmax = 8,
ylabel = {avg. \# of iterations},
]

\addplot[black,mark=o]
table[]{x y
0 50
1 36.333
2 21.514
3 11.968
4 8.6132
5 5.6229
6 4.5235
7 3.8004
8 3.3429
};\addlegendentry{LDPC, BP} 

\addplot[red,mark=x]
table[]{xx y
0 17.25
1 9.6667
2 6.1957
3 3.0974
4 1.7221
5 1.2431
6 0.87642
7 0.70183
8 0.57792
};\addlegendentry{dRM prod., SOGRAND} 

\addplot[blue,mark=x]
table[]{x y
0 17.458
1 14.167
2 7.1923
3 4.4857
4 2.2902
5 1.4518
6 1.0058
7 0.73643
8 0.63514
};\addlegendentry{dRM prod., ORBGRAND Pyndiah} 

\end{semilogyaxis}

\end{tikzpicture}
	\caption{Block Rician fading channel $(N_c=128,t=4,K=5)$ with \ac{QPSK} decoding performance of the $(1024, 441)$ 5G LDPC with maximum number of iterations $I_\text{max}=50$ as compared to a $(1024, 441)=(32,21)^2$ dRM product code decoded with \ac{SOGRAND}, i.e. using $1$line-ORBGRAND for SO list decoding $\alpha=0.5$ and maximum iteration number $I_\text{max}=20$, where lists are added to until $L=4$ or the predicted list-BLER is below $10^{-4}$.  Upper panel: \ac{BLER} and \ac{BER} performance. Middle panel: average number of queries per-bit until a decoding with \ac{SOGRAND}, where parallelized assumes all rows/columns are decoded in parallel. Lower panel: average number of iterations.}
	\label{fig:FER_fading1}
\end{figure}
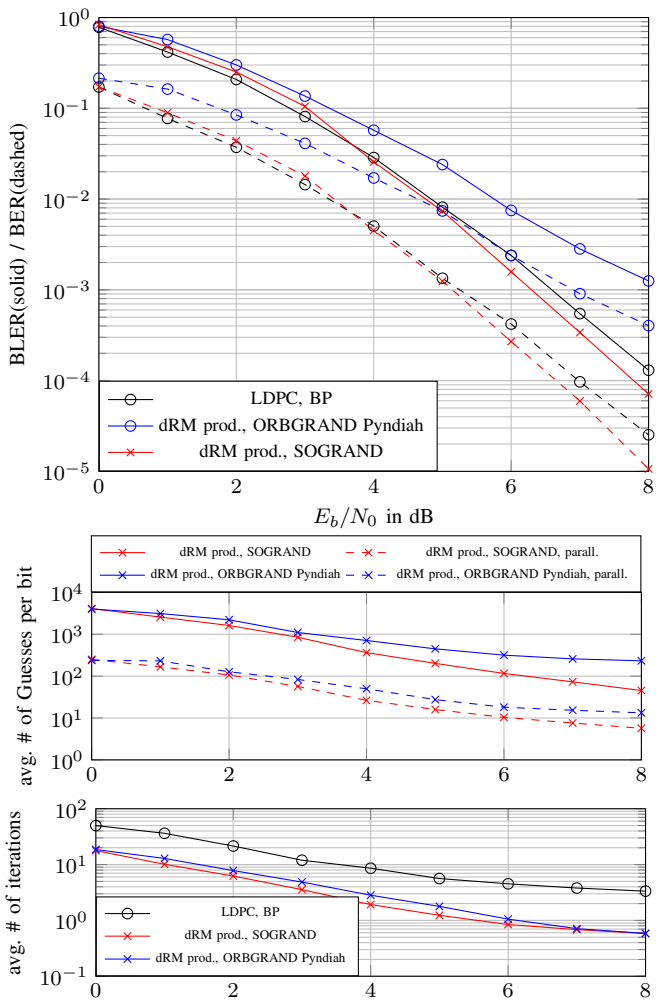

\begin{figure}
	\centering
	\footnotesize
	\begin{tikzpicture}[scale=1]
\footnotesize
\begin{semilogyaxis}[
legend style={at={(0,0)},anchor= south west, font=\scriptsize},
ymin=1e-6,
ymax=1,
width=3.5in,
height=3in,
grid=both,
xmin = 0,
xmax = 18,
xlabel = $E_b/N_0$ in dB,
ylabel = {BLER(solid) / BER(dashed)},
]

\addplot[black,mark=o]
table[]{x y
0 0.90909
2 0.44643
4 0.11601
6 0.022758
8 0.0055506
10 0.0013689
12 0.00056155
14 0.00023546
16 0.00011434
18 6.3117e-05
};\addlegendentry{LDPC, BP}

\addplot[blue,mark=o]
table[]{x y
0 0.95238
2 0.4878
4 0.16
6 0.057307
8 0.016129
10 0.0063492
12 0.002891
14 0.001325
16 0.00062619
18 0.00042841
};\addlegendentry{eBCH prod., ORBGRAND Pyndiah}

\addplot[red,mark=x]
table[]{x y
0 0.89286
2 0.66667
4 0.13089
6 0.020894
8 0.0052488
10 0.001035
12 0.00044941
14 0.00021093
16 0.00009687
18 4.9586e-05
};\addlegendentry{eBCH prod., SOGRAND} 

\addplot[black,mark=o,dashed,mark options=solid]
table[]{x y
0 0.1741
2 0.064191
4 0.01435
6 0.0033242
8 0.00080156
10 0.0001948
12 9.0861e-05
14 3.9171e-05
16 2.0828e-05
18 1.0186e-05
};

\addplot[blue,mark=o,dashed,mark options=solid]
table[]{x y
0 0.18703
2 0.093607
4 0.034438
6 0.012491
8 0.0036983
10 0.0014391
12 0.00066968
14 0.00032182
16 0.00014297
18 9.9363e-05
};

\addplot[red,mark=x,dashed,mark options=solid]
table[]{x y
0 0.12217
2 0.067201
4 0.012767
6 0.0019752
8 0.00050848
10 0.00010803
12 4.9523e-05
14 2.6402e-05
16 1.2832e-05
18 6.3435e-06
};

\end{semilogyaxis}

\end{tikzpicture}

\begin{tikzpicture}[scale=1]
\begin{semilogyaxis}[
legend style={at={(1,1)},anchor= south east, font=\tiny},
legend columns=2,
ymin=1e-3,
ymax=1e3,
width=3.5in,
height=1.5in,
ytick={10000,1000,100,10,1,0.1,0.01,0.001,0.0001,0.00001,0.000001,0.0000001,0.00000001},
grid=both,
xmin = 0,
xmax = 18,
ylabel = {avg. \# of Guesses per bit},
]

\addplot[red,mark=x]
table[]{x y
0 146.39
2 104.77
4 30.659
6 9.1122
8 2.5105
10 0.8296
12 0.26105
14 0.12549
16 0.07219
18 0.047852
};\addlegendentry{eBCH prod., SOGRAND} 

\addplot[red,mark=x,dashed,mark options=solid]
table[]{x y
0 9.0628
2 6.8918
4 2.0594
6 0.64072
8 0.2226
10 0.087621
12 0.028603
14 0.010847
16 0.0046217
18 0.0023605
};\addlegendentry{eBCH prod., SOGRAND, parall.}

\addplot[blue,mark=x]
table[]{x y
0 143.72
2 95.465
4 41.455
6 16.844
8 11.461
10 12.042
12 13.469
14 14.834
16 15.953
18 16.459
};\addlegendentry{eBCH prod., ORBGRAND Pyndiah} 

\addplot[blue,mark=x,dashed,mark options=solid]
table[]{x y
0 9.2945
2 5.9343
4 2.6752
6 1.0813
8 0.68959
10 0.63037
12 0.61724
14 0.6215
16 0.6363
18 0.64001
};\addlegendentry{eBCH prod., ORBGRAND Pyndiah, parall.} 

\end{semilogyaxis}

\end{tikzpicture}

\begin{tikzpicture}[scale=1]
\begin{semilogyaxis}[
legend style={at={(1,1)},anchor= north east, font=\tiny},
legend columns=1,
ymin=1e-1,
ymax=1e2,
width=3.5in,
height=1.5in,
ytick={10000,1000,100,10,1,0.1,0.01,0.001,0.0001,0.00001,0.000001,0.0000001,0.00000001},
grid=both,
xmin = 0,
xmax = 18,
ylabel = {avg. \# of iterations},
]

\addplot[black,mark=o]
table[]{x y
0 47
2 28.571
4 11.921
6 5.4965
8 3.2048
10 2.4382
12 2.1455
14 2.0459
16 2.0208
18 2.008
};\addlegendentry{LDPC, BP}

\addplot[red,mark=x]
table[]{x y
0 19.182
2 13.529
4 4.0235
6 1.3615
8 0.67619
10 0.55659
12 0.51629
14 0.50812
16 0.50394
18 0.50166
};\addlegendentry{eBCH prod., SOGRAND} 

\addplot[blue,mark=x]
table[]{x y
0 18.682
2 12.278
4 4.9554
6 1.7375
8 0.86875
10 0.66861
12 0.56843
14 0.52502
16 0.52322
18 0.51459
};\addlegendentry{eBCH prod., ORBGRAND Pyndiah} 

\end{semilogyaxis}

\end{tikzpicture}
	\caption{Block Rician fading channel $(N_c=256,t=2,K=8)$ with \ac{QPSK} decoding erformance of the $(1024,676)$ 5G LDPC with max.iterations $50$ compared to a $(1024,676)=(32,26)^2$ eBCH product code decoded with \ac{SOGRAND} with $\alpha=0.5$ and max iterations $20$, where lists are added to until $L=4$ or the predicted list-BLER is below $10^{-4}$. Upper panel: \ac{BLER} and \ac{BER}. Middle panel: average number of queries per-bit until a decoding with \ac{SOGRAND}, where parallelized assumes all rows/columns are decoded in parallel. Lower panel: average number of iterations.}
	\label{fig:FER_fading2}
\end{figure}

\begin{figure}
	\centering
	\footnotesize
	\begin{tikzpicture}[scale=1]
\footnotesize
\begin{semilogyaxis}[
legend style={at={(0,0)},anchor= south west, font=\scriptsize},
ymin=1e-7,
ymax=1,
width=3.5in,
height=3in,
grid=both,
xmin = 0,
xmax = 12,
xlabel = $E_b/N_0$ in dB,
ylabel = {BLER(solid) / BER(dashed)},
]

\addplot[black,mark=o]
table[]{x y
0 1
2 0.89286
4 0.21459
6 0.022894
8 0.0011978
10 9.5809e-05
12 2.1119e-05
};\addlegendentry{LDPC, BP}

\addplot[blue,mark=o]
table[]{x y
0 1
2 1
4 0.4878
6 0.10811
8 0.022173
10 0.0059154
12 0.0019635
};\addlegendentry{eBCH prod., ORBGRAND Pyndiah}

\addplot[red,mark=x]
table[]{x y
0 1
2 0.90909
4 0.27174
6 0.027056
8 0.0011023
10 7.3113e-5
12 1.4564e-05
};\addlegendentry{eBCH prod., SOGRAND} 

\addplot[black,mark=o,dashed,mark options=solid]
table[]{x y
0 0.15579
2 0.085477
4 0.013175
6 0.0014079
8 6.3166e-05
10 6.3637e-06
12 1.1726e-06
};

\addplot[blue,mark=o,dashed,mark options=solid]
table[]{x y
0 0.16577
2 0.12944
4 0.059523
6 0.013374
8 0.0028182
10 0.00074737
12 0.00024424
};

\addplot[red,mark=x,dashed,mark options=solid]
table[]{x y
0 0.12437
2 0.068173
4 0.013351
6 0.0013603
8 4.7962e-05
10 4.682e-06
12 7.4334e-07
};

\end{semilogyaxis}

\end{tikzpicture}

\begin{tikzpicture}[scale=1]
\begin{semilogyaxis}[
legend style={at={(1,1)},anchor= south east, font=\tiny},
legend columns=2,
ymin=1e-2,
ymax=1e3,
width=3.5in,
height=1.5in,
ytick={10000,1000,100,10,1,0.1,0.01,0.001,0.0001,0.00001,0.000001,0.0000001,0.00000001},
grid=both,
xmin = 0,
xmax = 12,
ylabel = {avg. \# of Guesses per bit},
]

\addplot[red,mark=x]
table[]{x y
0 153.58
2 144.32
4 66.925
6 14.754
8 3.7788
10 1.2767
12 0.43314
};\addlegendentry{eBCH prod., SOGRAND} 
\addplot[red,mark=x,dashed,mark options=solid]
table[]{x y
0 5.3619
2 5.1663
4 2.4483
6 0.55934
8 0.16844
10 0.082142
12 0.03634
};\addlegendentry{eBCH prod., SOGRAND, parall.}

\addplot[blue,mark=x]
table[]{x y
0 152.24
2 152.36
4 89.327
6 29.393
8 11.109
10 8.7744
12 9.8997
};\addlegendentry{eBCH prod., ORBGRAND Pyndiah} 

\addplot[blue,mark=x,dashed,mark options=solid]
table[]{x y
0 5.5368
2 5.5264
4 3.3215
6 1.0864
8 0.41249
10 0.30893
12 0.31702
};\addlegendentry{eBCH prod., ORBGRAND Pyndiah, parall.} 

\end{semilogyaxis}

\end{tikzpicture}

\begin{tikzpicture}[scale=1]
\begin{semilogyaxis}[
legend style={at={(0,0)},anchor= south west, font=\tiny},
ymin=1e-1,
ymax=1e2,
width=3.5in,
height=1.5in,
ytick={10000,1000,100,10,1,0.1,0.01,0.001,0.0001,0.00001,0.000001,0.0000001,0.00000001},
grid=both,
xmin = 0,
xmax = 12,
ylabel = {avg. \# of iterations},
]
\addplot[black,mark=o]
table[]{x y
0 50
2 50
4 21.227
6 7.4852
8 4.4503
10 3.5489
12 3.2044
};\addlegendentry{LDPC, BP}

\addplot[red,mark=x]
table[]{x y
0 20
2 19.333
4 8.9355
6 2.0685
8 0.76482
10 0.56362
12 0.51951
};\addlegendentry{eBCH prod., SOGRAND} 

\addplot[blue,mark=x]
table[]{x y
0 20
2 20
4 11.526
6 3.5559
8 1.1
10 0.65973
12 0.57074
};\addlegendentry{eBCH prod., ORBGRAND Pyndiah} 

\end{semilogyaxis}

\end{tikzpicture}
	\caption{Block Rician fading channel $(N_c=256,t=8,K=6)$ with \ac{QPSK} decoding performance of the $(4096,3249)$ 5G LDPC with max. iterations $50$ compared to a $(4096,3249)=(64,57)^2$ eBCH product code decoded with \ac{SOGRAND}, $\alpha=0.5$ and max. iteration number $20$, where lists are added to until $L=4$ or the predicted list-BLER is below $10^{-5}$.  Upper panel: \ac{BLER} and \ac{BER}. Middle panel: average number of queries per-bit until a decoding with \ac{SOGRAND}, where parallelized assumes all rows/columns are decoded in parallel. Lower panel: average number of iterations.}
	\label{fig:FER_fading3}
\end{figure}
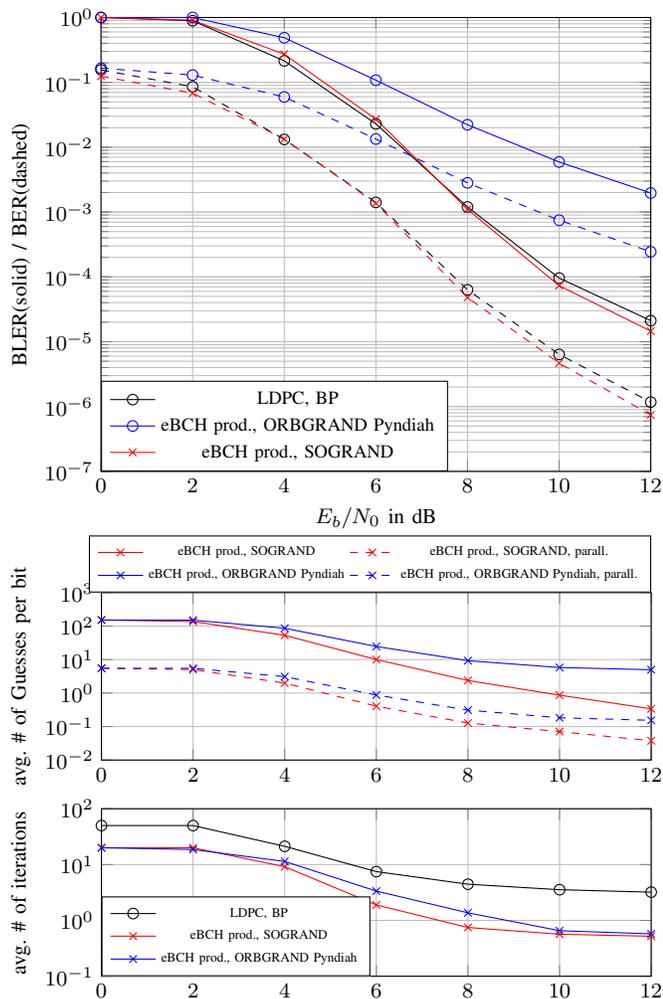

Performance evaluation results are shown in Fig.~\ref{fig:FER_fading1} and Fig.~\ref{fig:FER_fading2} 
for $1024$ bit codes constructed with dRM and eBCH codes, both of which \ac{GRAND} algorithms can readily decode. Fig.~\ref{fig:FER_fading3} shows equivalent results for a longer $4096$ bit code. The results in fading channels mirror those found for \ac{AWGN} channels, 
showing that the product codes, when using iterative \ac{SOGRAND}, significantly outperform Pyndiah's iterative scheme and have better \ac{BLER} and \ac{BER} performance compared to 5G \acp{LDPC} with \ac{BP} decoding, in a highly modularlized and parallelizable design that has a small number of iterations for decoding.

\subsection{Encoding Complexity.}
Product codes have extremely attractive properties in terms of encoding complexity and latency. To encode each row or column of a systematic code requires only $k(n-k)$ binary operations, where $(n-k)$ is small. There are $k$ such rows and $n$ columns, resulting in a total complexity of $(k+n)k(n-k)$, but all $k$ components can be encoded in parallel and then all $n$ components can be encoded in parallel, resulting in a latency of $2k(n-k)$ which is of order $k$. In hardware, for component codes of this size, encoding can be done in one clock cycle so, with full parallelism, encoding would take a total of two clock cycles. The \ac{QC} structure of the GLDPC codes we present also allows linear complexity encoding as for product codes~\cite{li2006efficient, ZP01} with similar parallelism for minimal latency.

In contrast, the encoding complexity of LDPC codes can be significant. The na\"ive application of Gaussian elimination schemes leads to cubic complexity in $n$. In order to tackle the complexity issues with LDPC encoding \cite{richardson2001efficient},
several structures have been imposed on LDPC codes. These constraints include LDPC codes with convolutional codes \cite{Chenetal06},
spatially coupled LDPC codes \cite{
Mitchelletal15}, using the sparsity of LDPC codes for lower triangulation in the encoding matrix \cite{RL14}, restricting the Tanner graphs to exhibit a more product-code-like structure \cite{
Huetal01}, column-scaled LDPC codes \cite{Zhaoetal13}, dual-diagonal structures \cite{Linetal08}, increasing field sizes \cite{Chang11, Songetal19}, root-protographs structures \cite{Fangetal19} and, primarily, \ac{QC}-LDPCs
 \cite{Djordjevicetal04}.
No such nuance is required with \ac{SOGRAND}.

\section{Discussion}

Efficient soft-detection decoding of powerful long, low-rate error correction codes has long  been a core objective in communications. The successful approach has been to create long codes by appropriate concatenation of component codes that provide \ac{SO} from their \ac{SI} in an iterative fashion. Turbo codes use powerful components but approximate \ac{SO} while \ac{LDPC} codes use weak components but accurate \ac{SO}. 

Here we demonstrate that \ac{GRAND} can bridge the gap between Turbo decoding of product codes and \ac{LDPC} codes by decoding powerful component codes and providing accurate SO with no additional computational burden. We have shown that when decoded with \ac{SOGRAND} simple product-like codes that avail of powerful, high-rate component codes can outperform the \ac{LDPC} codes in the 5G standard in fading channels. As \ac{GRAND} algorithms can efficiently decode any moderate redundancy component code, product codes offer a wide design space when decoded with \ac{SOGRAND}.

Our ultimate goal in this research direction is to design a comprehensive coding scheme that supports all \acp{MCS} in wireless systems. For example, for high-rate short codes, we directly use GRAND, while for longer codes, we apply SOGRAND with product or GLDPC codes. For higher-order modulations, product codes and GLDPC codes with SOGRAND are fully compatible with BICM, much like LDPC codes. In practical applications, punctured, shortened, or irregular product codes can be employed for rate adaptation and \ac{IR-HARQ}. For GLDPC codes, a Raptor-like structure~\cite{chen2015protograph}, akin to that used in 5G LDPC codes, could be applied for \ac{IR-HARQ}.

Circuit designs and a realization of \ac{ORBGRAND} in hardware~\cite{Riaz23} illustrate that codebook queries can be parallelised, resulting in low latency, low energy \ac{ORBGRAND} decoding of component codes. The component codes can be decoded in parallel to generate ultra low-latency long-code implementations. Due to that feature, it is possible to trade-off decoding area versus latency in hardware by determining the number of \ac{ORBGRAND} circuits in a chip. Moreover, the product-code-like design is highly modular and can be readily adapted to distinct lengths and rates without resorting to puncturing or significantly changing the architecture of the decoder. Product-like codes also have a significant encoding latency advantage over \ac{LDPC} codes as all rows and columns (or diagonals) can be encoded in parallel. 

While we demonstrate the approach with product and \ac{GLDPC} codes, a much broad palette of more sophisticated and ultimately more powerful, long, low-rate codes now becomes practical for soft detection decoding, offering a viable alternative to \ac{LDPC} codes as well as additional possibilities. For example codes such as implicit partial product-LDPC codes~\cite{Wangetal23}, which use \ac{LDPC} codes for rows and BCH codes for columns. 

\section*{Acknowledgments}
The authors would like to thank the Associate Editor and the anonymous reviewers for their valuable comments, which improved the presentation of the work.

\bibliographystyle{IEEEtran}
\bibliography{references}

\begin{thebibliography}{10}
\providecommand{\url}[1]{#1}
\csname url@samestyle\endcsname
\providecommand{\newblock}{\relax}
\providecommand{\bibinfo}[2]{#2}
\providecommand{\BIBentrySTDinterwordspacing}{\spaceskip=0pt\relax}
\providecommand{\BIBentryALTinterwordstretchfactor}{4}
\providecommand{\BIBentryALTinterwordspacing}{\spaceskip=\fontdimen2\font plus
\BIBentryALTinterwordstretchfactor\fontdimen3\font minus \fontdimen4\font\relax}
\providecommand{\BIBforeignlanguage}[2]{{%
\expandafter\ifx\csname l@#1\endcsname\relax
\typeout{** WARNING: IEEEtran.bst: No hyphenation pattern has been}%
\typeout{** loaded for the language `#1'. Using the pattern for}%
\typeout{** the default language instead.}%
\else
\language=\csname l@#1\endcsname
\fi
#2}}
\providecommand{\BIBdecl}{\relax}
\BIBdecl

\bibitem{galligan2023upgrade}
K.~Galligan, P.~Yuan, M.~M\'edard, and K.~R. Duffy, ``Upgrade error detection to prediction with {GRAND},'' in \emph{IEEE Globecom}, 2023.

\bibitem{shannon1948}
C.~E. Shannon, ``A mathematical theory of communication,'' \emph{The Bell System Technical Journal}, vol.~27, no.~3, pp. 379--423, 1948.

\bibitem{elias_error-free_1954}
P.~Elias, ``Error-free {Coding},'' \emph{Trans. IRE Prof. Group Inf. Theory}, vol.~4, no.~4, pp. 29--37, 1954.

\bibitem{forney1966_concatenated}
G.~D. Forney, \emph{{Concatenated Codes}}.\hskip 1em plus 0.5em minus 0.4em\relax {MIT Press}, 1966.

\bibitem{berrou_1993_turbo}
C.~Berrou, A.~Glavieux, and P.~Thitimajshima, ``{Near Shannon limit error-correcting coding and decoding: Turbo-codes},'' in \emph{IEEE ICC}, 1993.

\bibitem{pyndiah_1998}
R.~Pyndiah, ``Near-optimum decoding of product codes: block turbo codes,'' \emph{IEEE Trans. Commun.}, vol.~46, no.~8, pp. 1003--1010, 1998.

\bibitem{gallagher1962_lowdensity}
R.~Gallager, ``Low-density parity-check codes,'' \emph{IRE Trans. Inf. Theory}, vol.~8, pp. 21--28, 1962.

\bibitem{costello2007channel}
D.~J. Costello and G.~D. Forney, ``Channel coding: The road to channel capacity,'' \emph{Proc. IEEE}, vol.~95, no.~6, pp. 1150--1177, 2007.

\bibitem{Sipser96}
M.~Sipser and D.~Spielman, ``Expander codes,'' \emph{IEEE Trans. Inf. Theory}, vol.~42, no.~6, pp. 1710--1722, 1996.

\bibitem{mackay1997near}
D.~J. MacKay and R.~M. Neal, ``Near {Shannon} limit performance of low density parity check codes,'' \emph{Electron. Lett.}, vol.~33, no.~6, pp. 457--458, 1997.

\bibitem{mackay2003information}
D.~J. MacKay, \emph{Information theory, inference and learning algorithms}.\hskip 1em plus 0.5em minus 0.4em\relax Cambridge university press, 2003.

\bibitem{richardson2001efficient}
T.~J. Richardson and R.~L. Urbanke, ``Efficient encoding of low-density parity-check codes,'' \emph{IEEE Trans. Inf. Theory}, vol.~47, no.~2, pp. 638--656, 2001.

\bibitem{richardson2001capacity}
------, ``The capacity of low-density parity-check codes under message-passing decoding,'' \emph{IEEE Trans. Inf. Theory}, vol.~47, no.~2, pp. 599--618, 2001.

\bibitem{chung2001design}
S.-Y. Chung, G.~D. Forney, T.~J. Richardson, and R.~Urbanke, ``On the design of low-density parity-check codes within 0.0045 {dB} of the {Shannon} limit,'' \emph{IEEE Commun. Lett.}, vol.~5, no.~2, pp. 58--60, 2001.

\bibitem{mansour2003high}
M.~M. Mansour and N.~R. Shanbhag, ``High-throughput {LDPC} decoders,'' \emph{IEEE Trans. Very Large Scale Integr. VLSI Syst.}, vol.~11, no.~6, pp. 976--996, 2003.

\bibitem{hocevar2004reduced}
D.~E. Hocevar, ``A reduced complexity decoder architecture via layered decoding of {LDPC} codes,'' in \emph{IEEE SiPS}, 2004, pp. 107--112.

\bibitem{zhang2010efficient}
Z.~Zhang, V.~Anantharam, M.~J. Wainwright, and B.~Nikolic, ``An efficient {10GBASE-T} ethernet {LDPC} decoder design with low error floors,'' \emph{IEEE J. of Solid-State Circuits}, vol.~45, no.~4, pp. 843--855, 2010.

\bibitem{hailes2015survey}
P.~Hailes, L.~Xu, R.~G. Maunder, B.~M. Al-Hashimi, and L.~Hanzo, ``A survey of {FPGA}-based {LDPC} decoders,'' \emph{IEEE Commun. Surv. Tutor.}, vol.~18, no.~2, pp. 1098--1122, 2015.

\bibitem{kim2016low}
K.-J. Kim, S.~Myung, S.-I. Park, J.-Y. Lee, M.~Kan, Y.~Shinohara, J.-W. Shin, and J.~Kim, ``Low-density parity-check codes for {ATSC} 3.0,'' \emph{IEEE Trans. Broadcast.}, vol.~62, no.~1, pp. 189--196, 2016.

\bibitem{richardson2018design}
T.~Richardson and S.~Kudekar, ``Design of low-density parity check codes for {5G} new radio,'' \emph{IEEE Commun. Mag.}, vol.~56, no.~3, pp. 28--34, 2018.

\bibitem{duffy_capacity-achieving_2019}
K.~R. Duffy, J.~Li, and M.~Medard, ``Capacity-achieving guessing random additive noise decoding,'' \emph{IEEE Trans. Inf. Theory}, vol.~65, no.~7, pp. 4023--4040, 2019.

\bibitem{galligan2021}
K.~Galligan, A.~Solomon, A.~Riaz, M.~M\'edard, R.~T. Yazicigil, and K.~R. Duffy, ``{IGRAND: decode any product code},'' in \emph{IEEE Globecom}, 2021.

\bibitem{An22}
W.~An, M.~M\'edard, and K.~R. Duffy, ``Keep the bursts and ditch the interleavers,'' \emph{IEEE Trans. Commun.}, vol.~70, no.~6, pp. 3655--3667, 2022.

\bibitem{Duffy19a}
K.~R. Duffy, M.~M\'edard, and W.~An, ``Guessing random additive noise decoding with symbol reliability information ({SRGRAND}),'' in \emph{IEEE Trans. Commun.}, vol.~70, no.~1, 2022, pp. 3--18.

\bibitem{duffy2022_ordered}
K.~R. Duffy, W.~An, and M.~Medard, ``Ordered reliability bits guessing random additive noise decoding,'' \emph{IEEE Trans. Signal Proc.}, vol.~70, pp. 4528 -- 4542, 2022.

\bibitem{abbas2021list}
S.~M. Abbas, M.~Jalaleddine, and W.~J. Gross, ``{List-GRAND}: A practical way to achieve maximum likelihood decoding,'' \emph{IEEE Trans. Very Large Scale Integr. Syst.}, no.~1, pp. 43--54, 2022.

\bibitem{an2023soft}
W.~An, M.~M{\'e}dard, and K.~R. Duffy, ``Soft decoding without soft demapping with {ORBGRAND},'' in \emph{IEEE ISIT}, 2023, pp. 1080--1084.

\bibitem{Duffy23ORBGRANDAI}
K.~R. Duffy, M.~Grundei, and M.~M\'edard, ``Using channel correlation to improve decoding -- {ORBGRAND-AI},'' in \emph{IEEE Globecom}, 2023.

\bibitem{cohen2023aes}
A.~Cohen, R.~G. D’Oliveira, K.~R. Duffy, J.~Woo, and M.~M{\'e}dard, ``{AES} as error correction: Cryptosystems for reliable communication,'' \emph{IEEE Commun. Lett}, vol.~27, no.~8, pp. 1964--1968, 2023.

\bibitem{Riaz21}
A.~Riaz, V.~Bansal, A.~Solomon, W.~An, Q.~Liu, K.~Galligan, K.~R. Duffy, M.~M\'edard, and R.~T. Yazicigil, ``Multi-code multi-rate universal maximum likelihood decoder using {GRAND},'' in \emph{IEEE ESSCIRC}, 2021, pp. 239--246.

\bibitem{Riaz23}
A.~Riaz, A.~Yasar, F.~Ercan, W.~An, J.~Ngo, K.~Galligan, M.~M\'edard, K.~R. Duffy, and R.~T. Yazicigil, ``A sub-0.8p{J}/b 16.3{G}bps/mm$^2$ universal soft-detection decoder using {ORBGRAND} in 40nm {CMOS},'' in \emph{IEEE ISSCC}, 2023.

\bibitem{Burg24}
L.~D. Blanc, V.~Herrmann, Y.~Ren, C.~Müller, A.~T. Kristensen, A.~Levisse, Y.~Shen, and A.~Burg, ``A {GRANDAB} decoder with 8.48 {Gbps} worst-case throughput in 65nm {CMOS},'' in \emph{IEEE ESSERC}, 2024, pp. 685--688.

\bibitem{condo2021_highperformance}
C.~Condo, V.~Bioglio, and I.~Land, ``High-performance low-complexity error pattern generation for {ORBGRAND} decoding,'' in \emph{{IEEE Globecom}}, 2021.

\bibitem{abbas2020grand}
S.~M. Abbas, T.~Tonnellier, F.~Ercan, and W.~J. Gross, ``{High-Throughput VLSI Architecture for GRAND},'' in \emph{IEEE Workshop on Sig. Proc. Sys.}, 2020, pp. 681--693.

\bibitem{abbas2021orbgrand}
S.~M. Abbas, T.~Tonnellier, F.~Ercan, M.~Jalaleddine, and W.~J. Gross, ``{High-Throughput and Energy-Efficient VLSI Architecture for Ordered Reliability Bits GRAND},'' \emph{IEEE Trans. on VLSI Sys.}, vol.~30, no.~6, 2022.

\bibitem{condo2021fixed}
C.~Condo, ``A fixed latency {ORBGRAND} decoder architecture with {LUT}-aided error-pattern scheduling,'' \emph{IEEE Trans. Circuits Sys. I: Regular Papers}, vol.~69, no.~5, pp. 2203--2211, 2022.

\bibitem{condo2022_iterative}
------, ``Iterative soft-input soft-output decoding with ordered reliability bits {GRAND},'' in \emph{IEEE Globecom Workshops}, 2022, pp. 510--515.

\bibitem{galligan2023_block}
K.~Galligan, M.~Médard, and K.~R. Duffy, ``{Block turbo decoding with ORBGRAND},'' in \emph{CISS}, 2023.

\bibitem{Hadavian23}
R.~Hadavian, D.~Truhachev, K.~\mbox{El-Sankary}, H.~Ebrahimzad, and H.~Najafi, ``Ordered reliability direct error pattern testing \mbox{(ORDEPT)} algorithm,'' in \emph{IEEE GLOBECOM}, 2023.

\bibitem{forney1968_exponential}
G.~Forney, ``Exponential error bounds for erasure, list, and decision feedback schemes,'' \emph{{IEEE Trans. Inf. Theory}}, vol.~14, no.~2, pp. 206--220, 1968.

\bibitem{richardson2008modern}
T.~Richardson and R.~Urbanke, \emph{Modern coding theory}.\hskip 1em plus 0.5em minus 0.4em\relax Cambridge University Press, 2008.

\bibitem{CP21}
M.~C. Co\c{s}kun and H.~D. Pfister, ``An information-theoretic perspective on successive cancellation list decoding and polar code design,'' \emph{IEEE Trans. Inf. Theory}, vol.~68, no.~9, pp. 5779--5791, 2022.

\bibitem{smith2012staircase}
B.~P. {Smith}, A.~{Farhood}, A.~{Hunt}, F.~R. {Kschischang}, and J.~{Lodge}, ``Staircase {C}odes: {FEC} for 100 {G}b/s {OTN},'' \emph{J. Light. Technol.}, vol.~30, no.~1, pp. 110--117, 2012.

\bibitem{Liva08}
G.~Liva, W.~E. Ryan, and M.~Chiani, ``Quasi-cyclic generalized {LDPC} codes with low error floors,'' \emph{IEEE Trans. Commun.}, vol.~56, no.~1, pp. 49--57, 2008.

\bibitem{Lentmaier10}
M.~Lentmaier, G.~Liva, E.~Paolini, and G.~Fettweis, ``From product codes to structured generalized {LDPC} codes,'' in \emph{CHINACOM}, 2010.

\bibitem{Hof2010}
E.~Hof, I.~Sason, and S.~Shamai, ``Performance bounds for erasure, list, and decision feedback schemes with linear block codes,'' \emph{IEEE Trans. Inf. Theory}, vol.~56, no.~8, pp. 3754--3778, 2010.

\bibitem{3gpp.38.212}
3GPP, ``{NR; Multiplexing and channel coding},'' {3rd Generation Partnership Project (3GPP)}, Technical Specification (TS) 38.21, 2019, version 15.5.0.

\bibitem{sauter2023_error}
A.~Sauter, B.~Matuz, and G.~Liva, ``Error detection strategies for {CRC}-concatenated polar codes under successive cancellation list decoding,'' in \emph{CISS}, 2023.

\bibitem{niu2012crc}
K.~Niu and K.~Chen, ``{CRC}-aided decoding of {P}olar codes,'' \emph{IEEE Commun. Letters}, vol.~16, no.~10, pp. 1668--1671, 2012.

\bibitem{tal2015_list}
I.~Tal and A.~Vardy, ``{List Decoding of Polar Codes},'' \emph{IEEE Trans. Inf. Theory}, vol.~61, no.~5, pp. 2213--2226, 2015.

\bibitem{balatsoukas2015llr}
A.~Balatsoukas-Stimming, M.~B. Parizi, and A.~Burg, ``{LLR}-based successive cancellation list decoding of {P}olar codes,'' \emph{IEEE Trans. Signal Process.}, vol.~63, no.~19, pp. 5165--5179, 2015.

\bibitem{forney1973_viterbi}
G.~Forney, ``{The Viterbi algorithm},'' \emph{Proc. of the IEEE}, vol.~61, pp. 268--278, 1973.

\bibitem{raghavan1998_reliability}
A.~Raghavan and C.~Baum, ``{A reliability output Viterbi algorithm with applications to hybrid ARQ},'' \emph{IEEE Trans. Inf. Theory}, vol.~44, pp. 1214--1216, 1998.

\bibitem{hagenauer1989_viterbi}
J.~Hagenauer and P.~Hoeher, ``{A Viterbi algorithm with soft-decision outputs and its applications},'' in \emph{IEEE Globecom}, vol.~3, 1989, pp. 1680--1686.

\bibitem{yamamoto1980_viterbi}
H.~Yamamoto and K.~Itoh, ``Viterbi decoding algorithm for convolutional codes with repeat request,'' \emph{IEEE Trans. Inf. Theory}, vol.~26, no.~5, pp. 540--547, 1980.

\bibitem{lin_error_2004}
S.~Lin and D.~J. Costello, \emph{\BIBforeignlanguage{eng}{Error control coding: fundamentals and applications}}.\hskip 1em plus 0.5em minus 0.4em\relax Pearson/Prentice Hall, 2004.

\bibitem{rowshan2022_constrained}
M.~Rowshan and J.~Yuan, ``Constrained error pattern generation for {GRAND},'' in \emph{IEEE Int. Symp. on Inf. Theory}, 2022.

\bibitem{Liuetal23}
M.~Liu, Y.~Wei, Z.~Chen, and W.~Zhang, ``Orbgrand is almost capacity-achieving,'' \emph{IEEE Trans. Inf. Theory}, vol.~69, no.~5, pp. 2830--2840, 2023.

\bibitem{arikan2009}
E.~Arikan, ``Channel polarization: A method for constructing capacity-achieving codes for symmetric binary-input memoryless channels,'' \emph{IEEE Trans. Inf. Theory}, vol.~55, no.~7, pp. 3051--3073, 2009.

\bibitem{freudenberger2021_reduced}
J.~Freudenberger, D.~Nicolas~Bailon, and M.~Safieh, ``Reduced complexity hard- and soft-input {BCH} decoding with applications in concatenated codes,'' \emph{IET Circ. Device Syst.}, vol.~15, no.~3, pp. 284--296, 2021.

\bibitem{hashimoto1999_composite}
T.~Hashimoto, ``{Composite scheme LR+Th for decoding with erasures and its effective equivalence to Forney's rule},'' \emph{IEEE Trans. Inf. Theory}, vol.~45, no.~1, pp. 78--93, 1999.

\bibitem{liang2016hardware}
X.~Liang, J.~Yang, C.~Zhang, W.~Song, and X.~You, ``Hardware efficient and low-latency {CA-SCL} decoder based on distributed sorting,'' in \emph{IEEE Globecom}, 2016.

\bibitem{Tao_CASCL_21}
Y.~Tao, S.-G. Cho, and Z.~Zhang, ``A configurable successive-cancellation list polar decoder using split-tree architecture,'' \emph{IEEE J. Solid-State Circuits}, vol.~56, no.~2, pp. 612--623, 2021.

\bibitem{Kam_CASCL_22}
D.~Kam, B.~Y. Kong, and Y.~Lee, ``A 1.1us 1.56{G}b/s/mm2 cost-efficient large-list {SCL} polar decoder using fully-reusable {LLR} buffers in 28nm {CMOS} technology,'' in \emph{IEEE Symposium on VLSI Technology and Circuits}, 2022, pp. 204--205.

\bibitem{5271272}
R.~Steele and L.~Hanzo, \emph{Characterisation of Mobile Radio Channels}, 1999, pp. 91--185.

\bibitem{stuber2001principles}
G.~L. St{\"u}ber and G.~L. Steuber, \emph{Principles of mobile communication}.\hskip 1em plus 0.5em minus 0.4em\relax Springer, 2001, vol.~2.

\bibitem{li2006efficient}
Z.~Li, L.~Chen, L.~Zeng, S.~Lin, and W.~H. Fong, ``Efficient encoding of quasi-cyclic low-density parity-check codes,'' \emph{IEEE Trans. Commun.}, vol.~54, no.~1, pp. 71--81, 2006.

\bibitem{ZP01}
T.~Zhang and K.~Parhi, ``A class of efficient-encoding generalized low-density parity-check codes,'' in \emph{IEEE ICASSP}, vol.~4, 2001, pp. 2477--2480.

\bibitem{Chenetal06}
Z.~Chen, S.~Bates, D.~Elliott, and T.~Brandon, ``{CTH08-5}: Efficient encoding and termination of low-density parity-check convolutional codes,'' in \emph{IEEE Globecom}, 2006.

\bibitem{Mitchelletal15}
D.~G.~M. Mitchell, M.~Lentmaier, and D.~J. Costello, ``Spatially coupled {LDPC} codes constructed from protographs,'' \emph{IEEE Trans. Info. Theory}, vol.~61, no.~9, pp. 4866--4889, 2015.

\bibitem{RL14}
B.~Rajasekar and E.~Logashanmugam, ``Modified greedy permutation algorithm for low complexity encoding in {LDPC} codes,'' in \emph{ICCICCT}, 2014, pp. 336--339.

\bibitem{Huetal01}
X.-Y. Hu, E.~Eleftheriou, and D.-M. Arnold, ``Progressive edge-growth {T}anner graphs,'' in \emph{IEEE Globecom}, vol.~2, 2001, pp. 995--1001.

\bibitem{Zhaoetal13}
S.~Zhao, X.~Ma, X.~Zhang, and B.~Bai, ``A class of nonbinary {LDPC} codes with fast encoding and decoding algorithms,'' \emph{IEEE Trans. Commun.}, vol.~61, no.~1, pp. 1--6, 2013.

\bibitem{Linetal08}
C.-Y. Lin, C.-C. Wei, and M.-K. Ku, ``Efficient encoding for dual-diagonal structured {LDPC} codes based on parity bit prediction and correction,'' in \emph{APCCAS}, 2008, pp. 1648--1651.

\bibitem{Chang11}
N.~B. Chang, ``Rate adaptive non-binary {LDPC} codes with low encoding complexity,'' in \emph{ASILOMAR}, 2011, pp. 664--668.

\bibitem{Songetal19}
S.~Song, J.~Tian, J.~Lin, and Z.~Wang, ``A novel low-complexity joint coding and decoding algorithm for {NB-LDPC} codes,'' in \emph{IEEE ISCAS}, 2019.

\bibitem{Fangetal19}
Y.~Fang, P.~Chen, G.~Cai, F.~C. Lau, S.~C. Liew, and G.~Han, ``Outage-limit-approaching channel coding for future wireless communications: Root-protograph low-density parity-check codes,'' \emph{IEEE Veh. Technol. Mag.}, vol.~14, no.~2, pp. 85--93, 2019.

\bibitem{Djordjevicetal04}
I.~Djordjevic, S.~Sankaranarayanan, and B.~Vasic, ``Projective-plane iteratively decodable block codes for{ WDM} high-speed long-haul transmission systems,'' \emph{J. Light. Technol.}, vol.~22, no.~3, pp. 695--702, 2004.

\bibitem{chen2015protograph}
T.-Y. Chen, K.~Vakilinia, D.~Divsalar, and R.~D. Wesel, ``Protograph-based raptor-like ldpc codes,'' \emph{IEEE Trans. Commun.}, vol.~63, no.~5, pp. 1522--1532, 2015.

\bibitem{Wangetal23}
Y.~Wang, Q.~Wang, and X.~Ma, ``Design of implicit partial product-{LDPC} codes and low complexity decoding algorithm,'' \emph{IEEE Commun. Lett.}, vol.~27, no.~2, pp. 419--423, 2023.

\end{thebibliography}

\end{document}